\documentclass[a4paper,UKenglish,cleveref,thm-restate,numberwithinsect]{lipics-v2021}

\usepackage[T1]{fontenc}
\usepackage{graphicx}
\usepackage{xcolor}
\usepackage{mathtools,amsmath,amssymb}
\usepackage{complexity}
\usepackage{physics}
\usepackage{xspace}
\usepackage{csquotes}
\usepackage{bm}
\usepackage{cite}

\crefname{claim}{Claim}{Claims}
\Crefname{claim}{Claim}{Claims}
\crefname{numberedclaim}{Claim}{Claims}
\Crefname{numberedclaim}{Claim}{Claims}
\crefname{observation}{Observation}{Observations}
\Crefname{observation}{Observation}{Observations}

\newcommand{\normal}{\mathrm{norm}}
\newcommand{\conn}{\mathrm{conn}}
\newcommand{\compconn}{\mathrm{comp}}
\newcommand{\calE}{\mathcal{E}}
\newcommand{\calF}{\mathcal{F}}
\newcommand{\calO}{\mathcal{O}}
\newcommand{\calV}{\mathcal{V}}
\newcommand{\N}{\mathbb{N}}
\newcommand{\dotcup}{\mathbin{\mathaccent\cdot\cup}}
\newcommand{\lab}[1]{\textnormal{\texttt{#1}}}

\newcommand{\LSvar}[1]{L_{\mathrm{var}}(#1)}
\newcommand{\LSemb}[2]{L_{\mathrm{emb}}(#1,#2)}

\DeclareMathOperator{\skel}{skel}
\DeclareMathOperator{\pert}{pert}

\newcommand{\Cornuejols}{Cornu{\'e}jols\xspace}
\newcommand{\Lovasz}{Lov{\'{a}}sz\xspace}
\newcommand{\Sebo}{Seb{\H{o}}\xspace}

\newcommand{\probname}[1]{\textsc{#1}\xspace}
\newcommand{\GeneralizedFactor}{\probname{Generalized Factor}}
\newcommand{\MaximumFlow}{\probname{Maximum Flow}}

\newtheorem{question}[theorem]{Question}

\hideLIPIcs

\title{Recognition Complexity of Subgraphs of \texorpdfstring{$k$}{k}-Connected Planar Cubic Graphs}
\titlerunning{Recognizing \texorpdfstring{$\bm{k}$}{k}-Connected Subgraphs of Planar Cubic Graphs}

\author{Miriam~Goetze}{Institute of Theoretical Informatics, Karlsruhe Institute of Technology, Karlsruhe, Germany}
{miriam.goetze@kit.edu}
{https://orcid.org/0000-0001-8746-522X}{funded by the Deutsche Forschungsgemeinschaft (DFG, German Research Foundation) -- 520723789}

\author{Paul~Jungeblut}{Institute of Theoretical Informatics, Karlsruhe Institute of Technology, Karlsruhe, Germany}
{paul.jungeblut@kit.edu}
{https://orcid.org/0000-0001-8241-2102}{}

\author{Torsten~Ueckerdt}{Institute of Theoretical Informatics, Karlsruhe Institute of Technology, Karlsruhe, Germany}
{torsten.ueckerdt@kit.edu}
{https://orcid.org/0000-0002-0645-9715}{}

\authorrunning{M. Goetze, P. Jungeblut, and T. Ueckerdt}
\Copyright{Miriam Goetze, Paul Jungeblut, and Torsten Ueckerdt}


\relatedversion{Preliminary versions of this article appeared at the 30th Annual European Symposium on Algorithms (ESA 2022) and the 40th European Workshop on Computational Geometry (EuroCG 2024).} 
\relatedversiondetails[cite=Goetze2022-CubicSubgraphsESA]{Extended Abstract ESA 2022}{https://doi.org/10.4230/LIPIcs.ESA.2022.62} 
\relatedversiondetails[cite=Goetze2024-CubicSubgraphsEuroCG]{Abstract EuroCG 2024}{https://eurocg2024.math.uoi.gr/data/uploads/paper_40.pdf} 

\nolinenumbers


\begin{document}

\maketitle

\begin{abstract}
    We study the recognition complexity of subgraphs of $k$-con\-nect\-ed planar cubic graphs for~$k \in  \{1, 2, 3\}$.
    We present polynomial-time algorithms to recognize subgraphs of $1$- and $2$-connected planar cubic graphs, both in the variable and fixed embedding setting.
    The main tools involve the \textsc{Generalized (Anti)factor}-problem for the fixed embedding case, and SPQR-trees for the variable embedding case.
    
    Secondly, we prove \NP-hardness of recognizing subgraphs of $3$-connected planar cubic graphs in the variable embedding setting.
\end{abstract}

\section{Introduction}

Whether or not the \textsc{3-Edge-Colorability}-problem is solvable in polynomial time for planar graphs is one of the most fundamental open problems in algorithmic graph theory:

\begin{question}\label{quest:3-edge-colorability}
    Can we decide in polynomial time, whether the edges of a given planar graph can be colored in three colors such that any two adjacent edges receive distinct colors?
\end{question}

In other words, can we decide for a planar graph~$G$ in polynomial time whether $\chi'(G) \leq 3$, where~$\chi'(G)$ denotes the chromatic index of~$G$?
Clearly, it is enough to consider connected planar graphs~$G$ of maximum degree $\Delta(G) = 3$.
If~$G$ is connected, planar and $3$-regular, then by the Four-Color-Theorem~\cite{Appel1977_4Color1,Appel1977_4Color2} and the work of Tait~\cite{Tait1880_Bridgeless} we know that~$G$ is $3$-edge-colorable if and only if~$G$ is $2$-connected.
As we can check $2$-connectivity of subcubic graphs\footnote{A graph~$G$ is \emph{subcubic} if its maximum degree~$\Delta(G)$ is at most~$3$, and for such graphs $2$-vertex-connectivity and $2$-edge-connectivity are equivalent.} in linear time~\cite{Tarjan1974_Bridges}, we hence can decide in polynomial time whether a given $3$-regular planar graph is $3$-edge-colorable.

In particular, subgraphs of bridgeless $3$-regular planar graphs are $3$-edge-colorable.
However, this does not answer \cref{quest:3-edge-colorability} yet (as sometimes wrongly claimed, e.g., in~\cite{Cole2008_EdgeColoring}), because it is for example not clear which planar graphs of maximum degree~$3$ are subgraphs of $2$-connected $3$-regular planar graphs, and whether these can be recognized efficiently.

In this paper we consider the corresponding decision problem:
Given a graph~$G$, is there a $2$-connected $3$-regular planar graph~$H$, such that $G \subseteq H$?
In other words, can~$G$ be augmented, by adding edges and (possibly) vertices, to a supergraph~$H$ of~$G$ that is planar, $3$-regular, and $2$-connected?
For brevity, we call a planar, $3$-regular supergraph~$H$ a \emph{$3$-augmentation} of~$G$.
Motivated by \cref{quest:3-edge-colorability}, we are interested in $2$-connected $3$-augmentations of~$G$.

In fact, we shall consider the decision problems whether a given subcubic planar graph~$G$ admits a $k$-connected $3$-augmentation for each $k \in \{0,1,2,3\}$.
Equivalently, we study the recognition of subgraphs of $k$-connected cubic planar graphs. 
We consider several variants where the input graph~$G$ is given with a fixed embedding~$\calE$ and the desired $3$-augmentation~$H$ must extend~$\calE$, and/or where the input graph~$G$ is already $k'$-connected for some $k' \in \{0,1,2\}$.
Note that if~$G$ is $3$-connected, then~$H = G$ is the only connected~$3$-augmentation. 

\subparagraph*{Our Results.}
We resolve the complexity of finding a $k$-connected $3$-augmentation for a given subcubic planar graph~$G$ (with or without a given embedding), except when~$k=3$ and the embedding of~$G$ is given.
See also \cref{fig:overview} for an overview.

\begin{theorem}
    \label{thm:main_theorem}
    Let~$G$ be a planar graph with maximum degree~$\Delta(G) \leq 3$, let $n$ be the number of vertices of $G$, and let $\calE$ be an embedding of $G$.
    \begin{enumerate}
        \item\label{itm:main_theorem_1con_fixed}
        We can compute, in time~$\calO(n^2)$, a connected $3$-augmentation~$H$ extending~$\calE$, or conclude that none exists.
        \item \label{itm:main_theorem_1con_variable} We can compute, in time~$\calO(n)$, a connected $3$-augmentation~$H$ or conclude that none exists.
        \item\label{itm:main_theorem_2con_fixed}
        We can compute, in time~$\calO(n^4)$, a $2$-connected $3$-augmentation~$H$ extending~$\calE$, or conclude that none exists.
        If~$G$ is connected,~$\calO(n^2)$ time suffices.
        \item\label{itm:main_theorem_2con_variable} We can compute, in time~$\calO(n^2)$, a $2$-connected $3$-augmentation~$H$ or conclude that none exists.
        \item\label{itm:main_theorem_3con_variable}
        It is \NP-complete to decide whether~$G$ admits a $3$-connected $3$-augmentation, even if~$G$ is connected.
    \end{enumerate}
\end{theorem}

\noindent
Note that Statements~\ref{itm:main_theorem_1con_fixed} and~\ref{itm:main_theorem_2con_fixed} concern the fixed embedding setting, while Statements~\ref{itm:main_theorem_1con_variable},\ref{itm:main_theorem_2con_variable} and~\ref{itm:main_theorem_3con_variable} concern the variable embedding setting.

\definecolor{lightgreen}{rgb}{0.565, 0.933, 0.565}
\definecolor{gold}{rgb}{1, 0.843, 0}

\begin{figure}[htb]
    \centering
    \includegraphics[page=2]{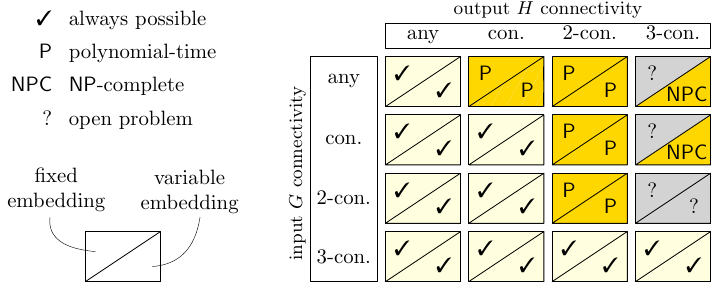}
    \caption{
        Complexity of finding $k$-connected $3$-augmentations (output connectivity $k \in \{0,1,2,3\}$) of $k'$-connected subcubic planar graphs (input connectivity $k' \in \{0,1,2\}$).
    }
    \label{fig:overview}
\end{figure}

Statement~\ref{itm:main_theorem_2con_variable} can be considered the main result of the paper.
Still, we emphasize that this does not answer \cref{quest:3-edge-colorability} yet.
In fact, admitting a $2$-connected $3$-augmentation is a sufficient condition for $3$-edge-colorability; but it is in general not necessary.
For example, $K_{2,3}$ admits a proper $3$-edge-coloring but no $2$-connected $3$-augmentation.
\Cref{quest:3-edge-colorability} remains open and we discuss it and its connection to $3$-augmentations in more detail in \cref{sec:discussion}.

In order to decide whether a given graph~$G$ admits a $3$-augmentation, we may of course assume that~$G$ itself is planar and of maximum degree at most~$3$.
Observe that it is always possible to find a (not necessarily connected) $3$-augmentation of~$G$, for example by adding the small gadget~$K_4^{(1)}$ consisting of~$K_4$ with one subdivided edge to each vertex that has not degree~$3$ yet, as illustrated in \cref{fig:examples_1-conn_aug}.

\begin{observation}
    \label{obs:connected_3-augmentation}
    Every subcubic planar graph~$G$ has a $3$-augmentation~$H$ extending its embedding. 
    If~$G$ is connected, then so is~$H$.
\end{observation}

However, this becomes non-trivial if we require the $3$-augmentation~$H$ to be $k$-connected for some $k \in \{1,2,3\}$.
See \cref{fig:examples} for some problematic cases.

\begin{figure}[t]
\centering
\begin{subfigure}[t]{0.3\textwidth}
    \centering
    \includegraphics[page=2]{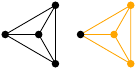}
    \caption{$G_1  \subseteq \textcolor{orange}{H_1}$}
    \label{fig:examples_0-conn_aug}
\end{subfigure}
\hspace{4mm}
\begin{subfigure}[t]{0.3\textwidth}
    \centering
    \includegraphics[page=3]{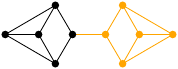}
    \caption{$G_2  \subseteq \textcolor{orange}{H_2}$}
    \label{fig:examples_1-conn_aug}
\end{subfigure}
\hspace{4mm}
\begin{subfigure}[t]{0.3\textwidth}
    \centering
    \includegraphics[page=4]{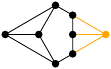}
    \caption{$G_3 \subseteq \textcolor{orange}{H_3}$}
    \label{fig:examples_2-conn_aug}
\end{subfigure}
    \centering
    \caption{
        For $k=1,2,3$, the planar subcubic graph~$G_k$ (in black) admits a $(k-1)$-con\-nect\-ed $3$-aug\-men\-ta\-tion~$H_k$ (new vertices and edges in orange), but no $k$-con\-nect\-ed $3$-augmentation.
    }
    \label{fig:examples}
\end{figure}

\subparagraph*{Previous Results.}
Hartmann, Rollin and Rutter~\cite{Hartmann2015_RegularAugmentation} studied, for each~$k, r \in \N$, whether a planar graph~$G$ can be augmented by adding edges (but no vertices!), to a $k$-connected $r$-regular planar graph~$H$.
In particular, for~$r = 3$, they show that the problem is \NP-complete in the variable embedding setting for all $k \in \{0,1,2,3\}$, as well as in the fixed embedding setting when~$k = 3$.
For the remaining cases of fixed embedding and $k \in \{0,1,2\}$ they present a polynomial-time algorithm.

\begin{remark}
    In fact, several of their concepts and techniques \cite{Hartmann2015_RegularAugmentation} are very similar to ours.
    In case of $G$ having a fixed embedding $\calE$, any $3$-augmentation $H$ (with or without new vertices) extending $\calE$ induces an assignment of each new edge $e$ that is incident to an ``old'' vertex of~$G$ to the face of $\calE$ that contains $e$.
    New edges at vertices of $G$ are called \emph{free valencies} \cite{Hartmann2015_RegularAugmentation}.
    
    The authors \cite{Hartmann2015_RegularAugmentation} present conditions of this assignment that are necessary and sufficient for a connected or $2$-connected $3$-augmentation \emph{without new vertices}. 
    (These conditions also allow for a polynomial-time algorithm to find such an assignment.)
    Their \emph{matching condition} and \emph{planarity condition} become obsolete in our setting.
    However, their \emph{connectivity condition} and \emph{biconnectivity condition} demand roughly twice as many free valencies to be assigned to a face with several connected components.
    (Intuitively, these components must be strung together in~\cite{Hartmann2015_RegularAugmentation}, while we can do a star-like connection.)
    Most crucially, their \emph{parity condition}, which requires the number of free valencies assigned to each face to be even, is no longer necessary nor sufficient in our setting.
    It is for example violated in every example in \cref{fig:examples}.
    
    Moreover, one might be tempted to find a connected or $2$-connected $3$-augmentation in our setting (with new vertices allowed), with fixed embedding, by preprocessing the input by inserting extra vertices, so as to always fulfil the parity condition (and then handle the connectivity or biconnectivity condition somehow).
    A reasonable attempt would be to subdivide some edges in the input graph.
    However, this might turn a No-instance into a Yes-instance, as shown for example in \cref{fig:problematic-example}.

    \begin{figure}[ht]
        \centering
        \begin{subfigure}[t]{0.4\textwidth}
            \centering
            \includegraphics[page=2]{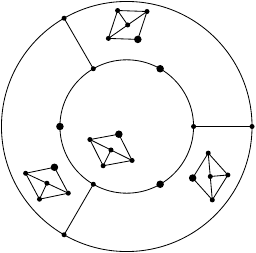}
            \caption{}
            \label{fig:problematic-example_left}
        \end{subfigure}
        \hspace{4mm}
        \begin{subfigure}[t]{0.4\textwidth}
            \centering
            \includegraphics[page=3]{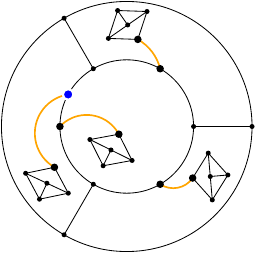}
            \caption{}
            \label{fig:problematic-example_right}
        \end{subfigure}
        \caption{\subref{fig:problematic-example_left} A graph $G$ with no connected $3$-augmentation extending its embedding.
        \subref{fig:problematic-example_right} After adding an extra degree-$2$ vertex (blue) to $G$, there is a connected $3$-augmentation.}
        \label{fig:problematic-example}
    \end{figure}

    Finally, for $k \in \{0,1,2\}$, finding a $k$-connected $3$-augmentation in the variable embedding setting is in \P~\cite{Goetze2022-CubicSubgraphsESA}, while the version without new vertices is \NP-complete~\cite{Hartmann2015_RegularAugmentation}.
    So to summarize, there is probably no direct reduction between the problem of augmenting by only adding edges and the problem of augmenting by adding vertices and edges.\hfill \qedsymbol
\end{remark}

Let us mention a few more examples from the rich and diverse area of augmentation problems.
Eswaran and Tarjan~\cite{Eswaran1976_AugmentationProblems} pioneered the systematic investigation of augmentation problems.
They presented algorithms to find in $\mathcal{O}(\abs{V(G)}+\abs{E(G)})$ a smallest number of edges whose addition to a given (not necessarily planar) graph $G$ results in a $2$-connected respectively $2$-edge-connected graph, while the weighted versions of either problem are \NP-complete.
If we additionally require the result to be planar, already both unweighted problems are \NP-complete, the same holds for the fixed embedding setting ~\cite{Kant1991_PlanarAugmentation,Rutter2012_Wolff_2EdgePlanarAugmentation}.
Other problems of augmenting to a planar graph consider augmenting to a grid graph~\cite{Bhatt1987_GridSubgraph}, or triangulating while minimizing the maximum degree~\cite{Kant1997_TriangulatingMaxDegree,Fraysseix1994_Augmentation}, avoiding separating triangles~\cite{Biedl1997_4ConnectedTriangulation}, creating a Hamiltonian cycle~\cite{DiGiacomo2010_HamiltonianAugmentation}, or resulting in a chordal graph~\cite{Kratochvil2012_Planar3Tree}, just to name a few.

\section{Preliminaries}
\label{sec:preliminaries}

All graphs considered here are finite, undirected, and contain no loops but possibly multi-edges.
We denote the degree of a vertex~$v$ in a graph~$G$ by~$\deg_G(v)$, the minimum degree in~$G$ by~$\delta(G)$, and the maximum degree by~$\Delta(G)$.
A graph~$G$ is \emph{$d$-regular}, for some non-negative integer~$d$, if we have $\delta(G) = \Delta(G) = d$.

A graph~$G$ is \emph{$k$-connected} if $G - S$ is connected for every set~$S \subseteq V(G)$ of at most~$k-1$ vertices in~$G$.
Similarly,~$G$ is \emph{$k$-edge-connected} if $G - S$ is connected for every set~$S \subseteq E(G)$ of at most~$k-1$ edges in~$G$.
We denote by~$\theta(G)$ the largest~$k$ for which~$G$ is $k$-edge-connected.
If~$G$ has maximum degree at most~$3$, then $G$ is $k$-connected if and only if~$\theta(G) \geq k$.
A \emph{bridge} in a graph~$G$ is an edge~$e$ whose removal increases the number of connected components, i.e., $G-e$ has strictly more components than~$G$.
Equivalently, $e$ is a bridge if~$e$ is not contained in any cycle of~$G$.
A \emph{bridgeless} graph is one that contains no bridge.
Note that a bridgeless graph may be disconnected.
Yet, for connected graphs of maximum degree~$3$, being bridgeless and being $2$-connected is equivalent.


A \emph{planar embedding}~$\calE$ of a (planar) graph $G$ is (in a sense that we need not make precise here) an equivalence class of crossing-free drawings of~$G$ in the plane.
In particular, a planar embedding determines the set~$F$ of all \emph{faces}, the distinguished \emph{outer face}, the clockwise ordering of incident edges around each vertex and the \emph{boundary} of each face as a set of \emph{facial walks}, each being a clockwise ordering of vertices and edges (with repetitions allowed).
The edges and vertices incident to the outer face are called \emph{outer edges} and \emph{outer vertices}, while all others are \emph{inner edges} and \emph{inner vertices}.
For every embedding~$\calE$ of~$G$ we define the \emph{flipped embedding~$\calE'$} to be the embedding obtained from~$\calE$ by reversing the clockwise order of incident edges at each vertex.
This operation changes neither the set of faces nor the outer face.
Whitney's Theorem~\cite{Whitney1933_UniqueEmbedding} states that a $3$-connected planar graph~$G$ has a unique embedding (up to the choice of the outer face and flipping).

\subparagraph*{Generalized Factors.}
Let~$H$ be a graph with a set $B(v) \subseteq \{0, \ldots, \deg_H(v)\}$ assigned to each vertex~$v \in V(H)$.
Following \Lovasz, a spanning subgraph~$G \subseteq H$ is called a \emph{$B$-factor} of~$H$ if and only if $\deg_G(v) \in B(v)$ for every vertex~$v \in V(H)$~\cite{Lovasz1972_Factorization}.
Deciding whether a graph~$H$ admits a $B$-factor is known as the \GeneralizedFactor problem.
In general, the \GeneralizedFactor problem is \NP-complete~\cite{Lovasz1972_Factorization}.
Still, for certain well-behaved sets~$B(\cdot)$, the problem becomes polynomial-time solvable.
A set~$B(v)$ is said to have a \emph{gap of length~$\ell \geq 1$} if there is an integer $i \in B(v)$ such that $i + 1, \ldots, i + \ell \notin B(v)$, and $i + \ell + 1 \in B(v)$.

\begin{theorem}[{\Cornuejols~\cite[Section~$3$]{Cornuejols1988_GeneralFactors}}]
\label{thm:generalized-factor-cornuejols}
Let~$H$ be a graph with a set $B(v) \subseteq \{0, \ldots, \deg_H(v)\}$ assigned to each vertex $v \in V(H)$.
If all gaps of each~$B(v)$ have length~$1$, then a $B$-factor can be computed in time~$\calO\bigl(\abs{V(H)}^4\bigr)$.
\end{theorem}

We say that there are \emph{no two consecutive forbidden degrees} for a vertex~$v \in V(H)$ if for all~$i, i+1 \in \{0, \dots, \deg_H(v)\}$ we have $i \in B(v)$ or $i+1 \in B(v)$.
Under this slightly stronger condition, an algorithm by \Sebo yields a better runtime for the \GeneralizedFactor problem.

\begin{theorem}[{\Sebo~\cite[Section~$3$]{Sebo1993_Antifactors}}]
    \label{thm:generalized-factor-sebo}
    Let~$H$ be a graph with a set $B(v) \subseteq \{0, \ldots, \deg_H(v)\}$ assigned to each vertex $v \in V(H)$.
    If no two consecutive degrees are forbidden for any vertex,
    then we can compute a $B$-factor in time $\mathcal{O}(|V(H)| \cdot |E(H)|)$, or conclude that no such exists.
\end{theorem}

\subparagraph{SPQR-Tree.}
The \emph{SPQR-tree} is a tree-like data structure that compactly encodes all planar embeddings of a $2$-connected planar graph.
It was introduced by Di Battista and Tamassia~\cite{DiBattista1996_SPQR} and can be computed in linear time~\cite{Gutwenger2001_SPQRLinear}.
Its precise definition includes quite a number of technical terms, of which we define the crucial ones below.
This makes our exposition self-contained, while also ensuring the established terminology for experienced readers.
We give an illustrating example in \cref{fig:SPQR-example}.

The SPQR-tree of a $2$-connected planar graph $G$ is a rooted tree $T$, where each vertex~$\mu$ of~$T$ is associated to a multigraph~$\skel(\mu)$ that is called the \emph{skeleton of~$\mu$}.
This multigraph $\skel(\mu)$ must be of one of four types determining whether~$\mu$ is an S-, a P-, a Q- or an R-vertex:
\begin{itemize}
    \item S-vertex: $\skel(\mu)$ is a simple cycle.
    \item P-vertex: $\skel(\mu)$ consists of two vertices and at least three parallel edges.
    \item Q-vertex: $\skel(\mu)$ consists of two vertices with two parallel edges.
    \item R-vertex: $\skel(\mu)$ is $3$-connected.
\end{itemize}
Some of the edges of the skeletons can be marked as \emph{virtual} edges.
An edge~$e = \mu\nu$ of the SPQR-tree~$T$ corresponds to two virtual edges, exactly one in~$\skel(\mu)$ and one in~$\skel(\nu)$.
Conversely, each virtual edge corresponds to exactly one tree edge of~$T$ in this way.
We refer again to \cref{fig:SPQR-example} for an example.

Under above conditions, the defining property of the SPQR-tree~$T$ is that~$G$ can be obtained by \emph{gluing} along the virtual edges:
For each tree edge~$e = \mu\nu$, the skeletons~$\skel(\mu)$ and~$\skel(\nu)$ are identified at the corresponding endpoints of the two virtual edges associated to~$e$, afterwards the virtual edges are removed.

We additionally require that no two S-vertices and no two P-vertices are adjacent in~$T$, as otherwise the skeletons of two such vertices can be merged into the skeleton of a new vertex of the same type.
Further, exactly one of the two parallel edges in a Q-vertex is a virtual edge while S-, P- and R-vertices contain only virtual edges.
Under these conditions the SPQR-tree of~$G$ is unique.
There is exactly one Q-vertex per edge in~$G$ and these form the leaves of  the SPQR-tree.
The inner S-, P- and R-vertices correspond more or less\footnote{
    In fact they correspond to so-called \emph{split pairs}.
    However, we omit their formal discussion, as it is not needed here.
} to the separation pairs (that is, pairs of vertices forming a cut set) of~$G$~\cite{DiBattista1996_SPQR}.

Assume that an arbitrary vertex~$\rho$ of~$T$ is fixed as the root.
For some vertex~$\mu$ in~$T$ let~$\pi$ be its parent.
Further, let~$u,v$ be the endpoints of the virtual edge in~$\skel(\mu)$ associated with the tree edge~$\mu\pi$ in~$T$.
Then the graph obtained by gluing~$\skel(\mu)$ with all skeletons in its subtree and without the virtual edge~$uv$ is called the \emph{pertinent graph} of~$\mu$ and denoted by~$\pert(\mu)$.
Note that~$\pert(\mu)$ is always connected.

\begin{figure}[tb]
    \centering
    \begin{subfigure}[t]{0.2\textwidth}
        \centering
        \includegraphics[page=2]{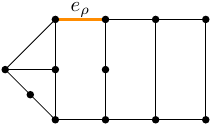}
        \caption{}
        \label{fig:SPQR-example-graph}
    \end{subfigure}
    \hspace{10mm}
    \begin{subfigure}[t]{0.7\textwidth}
        \centering
        \includegraphics[page=3, scale=0.9]{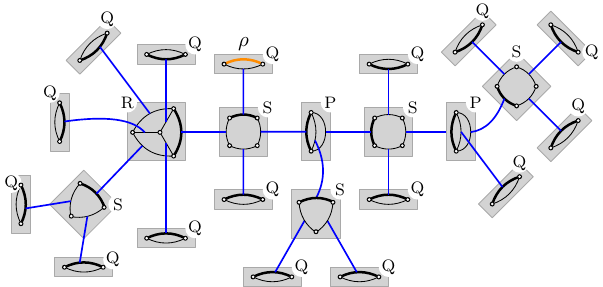}
        \caption{}
        \label{fig:SPQR-example-SPQR-tree}
    \end{subfigure}
    \caption{
        A graph with an edge~$e_\rho$ \subref{fig:SPQR-example-graph} and its SPQR-tree rooted at the Q-vertex~$\rho$ corresponding to~$e_\rho$ \subref{fig:SPQR-example-SPQR-tree}.
        Each tree node~$\mu$ shows the skeleton~$\skel(\mu)$ in which the virtual edge to its parent is shown thicker.
        The (blue) tree edges indicate the associated pairs of virtual edges.
    }
    \label{fig:SPQR-example}
\end{figure}

\subparagraph{SPQR-Tree and Planar Embeddings.}
If the SPQR-tree~$T$ is rooted at a Q-vertex~$\rho$ corresponding to an edge~$e_\rho$ of~$G$, then~$T$ represents all planar embeddings of~$G$ in which~$e_\rho$ is an outer edge~\cite{DiBattista1996_SPQR}.
When~$G$ is constructed by gluing corresponding virtual edges, one has the following choices on the planar embedding:
\begin{itemize}
    \item Whenever the corresponding virtual edges of an S-, P- or R-vertex~$\mu$ and its parent are glued together, this leaves two choices for the planar embedding:
    Having decided for an embedding~$\calE_\mu$ of $\pert(\mu)$ already, we can insert~$\calE_\mu$ or the flipped embedding~$\calE'_\mu$.
    \item The parallel virtual edges of a P-vertex~$\mu$ associated to virtual edges of children can be permuted arbitrarily.
    Every permutation leads to a different planar embedding of~$\skel(\mu)$.
    \item Gluing at the virtual edge of a Q-vertex~$\mu$ replaces the virtual edge~$uv$ by the ``real'' edge~$uv$ in~$G$.
    This has no effect on the embedding.
\end{itemize}
Let~$\calE$ be a planar embedding of~$G$ having~$e_\rho$ as an outer edge.
Further, let~$\mu$ be an inner vertex of the SPQR-tree and~$u_\mu, v_\mu$ be the endpoints of the virtual edge in~$\skel(\mu)$ corresponding to the parent edge of~$\mu$ in~$T$.
Lastly, let~$\calE_{\mu}$ be the restriction of~$\calE$ to~$\pert(\mu)$ and let~$f_{\mu}^o$ be the outer face of~$\calE_{\mu}$.
As~$e_\rho$ is an outer edge of~$\calE$, it follows that~$u_\mu$ and~$v_\mu$ are outer vertices in~$\calE_\mu$.
The $u_\mu v_\mu$-path in~$\pert(\mu)$ having~$f_{\mu}^o$ to its left (right) is the \emph{left (right) outer path} of~$\calE_{\mu}$.
Lastly, we define the \emph{left (right) outer face} of~$\calE_{\mu}$ inside~$\calE$ to be the face of~$\calE$ left (right) of the left (right) outer path of~$\calE_{\mu}$.

\section{2-Connected 3-Augmentations for a Fixed Embedding}
\label{sec:2con_3aug_fixed_embedding}

We consider the $3$-augmentation problem for arbitrary input graphs~$G$ and $2$-con\-nect\-ed output graphs~$H$, corresponding to the third column of the table in \cref{fig:overview}.
For the fixed embedding setting here, we present a quartic-time algorithm which yields the following.

\begin{theorem}
    \label{thm:2con_fixed}
    Let~$G$ be a planar $n$-vertex graph with $\Delta(G) \leq 3$ and an embedding~$\calE$.
    We can compute, in time~$\calO(n^4)$, a $2$-connected $3$-augmentation~$H$ of~$G$ extending~$\calE$, or conclude that none exists.
    If~$G$ is connected, then time $\calO(n^2)$ suffices.
\end{theorem}

This corresponds to Statement~\ref{itm:main_theorem_2con_fixed} of \cref{thm:main_theorem}.
We start with a reduction to graphs~$G$ with~$\delta(G) \geq 2$.

\begin{lemma}
    \label{lem:preprocessing}
    Let~$G$ be a planar graph with embedding~$\calE$.
    There is a planar su\-per\-graph~$G' \supseteq G$ with~$\delta(G') \geq 2$ whose embedding~$\calE'$ extends~$\calE$, such that $G$ has a $2$-connected $3$-augmentation extending~$\calE$ if and only if~$G'$ has one extending~$\calE'$.
\end{lemma}

\begin{proof}
    Consider the following two replacement rules, also shown in \cref{fig:preprocessing_deg0}--\subref{fig:preprocessing_deg1}:
    Each isolated vertex is replaced by a copy of~$K_3$, and each vertex~$v$ of degree~$1$ is replaced by a copy of~$K_3$ with one vertex connected to the other neighbor of~$v$.
    Let~$G'$ be the obtained graph such that its planar embedding~$\calE'$ extends~$\calE$.

    Let~$H$ be a $2$-connected $3$-augmentation of~$G$.
    We obtain a $2$-connected $3$-augmentation of~$G'$ as follows:
    For each vertex~$v$ of degree~$0$ (or~$1$) in~$G$, let~$N(v)$ be its three (two) new neighbors in~$H$.
    In~$H$, replace $v$ by its corresponding copy of~$K_3$.
    Connect its three (two) degree-$2$-vertices with one vertex of~$N(v)$ such that the embedding remains planar.

    The other direction works similar:
    In a $2$-connected $3$-augmentation of~$G'$, contract each copy of~$K_3$ that was introduced for a vertex~$v$ of~$G$ into a single vertex.
    If this creates multi-edges, replace each duplicated edge by the gadget shown in \cref{fig:preprocessing_parallel} to obtain a simple graph.
\end{proof}

\begin{figure}[htb]
    \centering
    \begin{subfigure}[t]{0.3\textwidth}
        \centering
        \includegraphics[page=2]{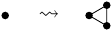}
        \caption{}
        \label{fig:preprocessing_deg0}
    \end{subfigure}
    \hfill
    \begin{subfigure}[t]{0.3\textwidth}
        \centering
        \includegraphics[page=3]{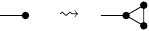}
        \caption{}
        \label{fig:preprocessing_deg1}
    \end{subfigure}
    \hfill
    \begin{subfigure}[t]{0.3\textwidth}
        \centering
        \includegraphics[page=4]{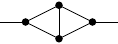}
        \caption{}
        \label{fig:preprocessing_parallel}
    \end{subfigure}
    \caption{
        \subref{fig:preprocessing_deg0}--\subref{fig:preprocessing_deg1} Replacement rules.
        \subref{fig:preprocessing_parallel} Gadget to avoid parallel edges.}
    \label{fig:preprocessing}
\end{figure}

\begin{lemma}
    \label{lem:2con_fixed_mindeg2}
    Let~$G$ be a planar $n$-vertex graph with an embedding~$\calE$, $\delta(G) \geq 2$, and $\Delta(G) \leq 3$.
    Then we can compute, in time~$\calO(n^4)$, a $2$-connected $3$-augmentation~$H$ of~$G$ extending~$\calE$, or conclude that none exists.
    If~$G$ is connected, then time $\calO(n^2)$ suffices.
\end{lemma}

\begin{proof}
    The proof is by a linear-time reduction to an equivalent instance~$A$ of the \GeneralizedFactor problem, such that~$A$ fulfills the necessary condition to apply an $\calO(n^4)$-time algorithm by \Cornuejols~(\cref{thm:generalized-factor-cornuejols}), or even an $\calO(n^2)$-time algorithm by \Sebo~(\cref{thm:generalized-factor-sebo}).

    We construct the $2$-connected $3$-augmentation~$H$ of~$G$ by adding new edges and vertices into the faces of~$\calE$.
    Therefore, the obtained embedding of~$H$ extends~$\calE$.
    
    Some faces of~$\calE$ stand out, as these \emph{must} contain new edges (and possibly vertices) to reach $2$-connectedness.
    We call these the \emph{connecting faces} and denote the set of connecting faces by~$F_{\conn}$.
    Obviously, all faces incident to at least two connected components are connecting faces.
    Further, for each bridge~$e$ of~$G$, the unique face~$f$ incident to both sides of~$e$ is a connecting face because the only way to add new connections between the components separated by~$e$ is through~$f$.
    Recall that a $3$-regular graph is $2$-connected if and only if it is connected and bridgeless, so these are the only two types of connecting faces.
    All other faces are considered to be \emph{normal} faces, denoted by~$F_{\normal}$.

    For a connecting face~$f \in F_{\conn}$, let~$G_f$ be the subgraph of~$G$ on the vertices and edges incident to~$f$, let $B_f$ be the set of its blocks (i.e., maximal $2$-connected components or bridges), and let~$T_f$ be its block-cut-forest.
    We partition~$B_f$ into~$S_f \cup I_f \cup L_f$, where we call the elements of~$S_f$ the \emph{singleton blocks}, the elements of~$I_f$ the \emph{inner blocks}, and the elements of~$L_f$ the \emph{leaf blocks}:
    \begin{align*}
        S_f &\coloneqq \{b \in B_f \mid \text{$b$ forms a trivial (i.e., single-vertex) tree in~$T_f$}\} \\
        I_f &\coloneqq \{b \in B_f \mid \text{$b$ is an inner vertex of a non-trivial tree in~$T_f$} \} \\
        L_f &\coloneqq \{b \in B_f \mid \text{$b$ is a leaf in a non-trivial tree in~$T_f$} \}
    \end{align*}
    In fact, we distinguish these types of blocks, since in order to obtain a $2$-connected $3$-augmentation, every singleton block must be incident to at least two new edges, and every leaf block has to be incident to at least one new edge.
    
    The \GeneralizedFactor instance~$A$ is a bipartite graph with bipartition classes~$\calV$ and~$\calF$.
    Here, $\calV \coloneqq \{v \in V(G) \mid \deg_G(v) = 2\}$ contains all vertices of~$G$ having degree lower than~$3$.
    Vertices in~$\calF$ represent the faces of~$\calE$.
    Edges of a $B$-factor of~$A$ will determine the faces of~$\calE$ containing the new edges.
    In particular, $\calF$ contains one vertex corresponding to each normal face in~$F_{\normal}$.
    Additional vertices in~$\calF$ are needed to handle the connecting faces.
    For each connecting face~$f \in F_{\conn}$, we add all blocks in~$B_f$ as vertices to~$\calF$.
    (If there are two faces~$f,g$ in~$\calE$ such that~$B_f$ and~$B_g$ contain blocks corresponding to the same subgraph of~$G$, then~$\calF$ contains two such vertices: one corresponding to the block in~$B_f$, and another to the block in~$B_g$. This is for example the case when~$G$ corresponds to a cycle.)

    In~$A$, each~$x \in \calF$ is incident to exactly the following~$v \in \calV$:
    If~$x$ represents a normal face~$f \in F_{\normal}$, then~$x$ is connected to every~$v \in \calV$ that is incident to~$f$ in~$\calE$.
    Otherwise, if~$x$ represents a block~$b \in B_f$ for some connecting face~$f \in F_{\conn}$, then~$x$ is connected to every~$v \in \calV$ that is contained in~$b$.
    See \cref{fig:generalized_factor_instance} for an example.

    \begin{figure}[tb]
        \centering
        \begin{subfigure}[tb]{0.25\textwidth}
            \centering
            \includegraphics[page=2]{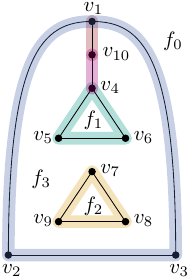}
            \caption{}
            \label{fig:generalized_factor_instance_graph}
        \end{subfigure}
        \hspace{4mm}
        \begin{subfigure}[tb]{0.3\textwidth}
            \centering
            \includegraphics[page=3]{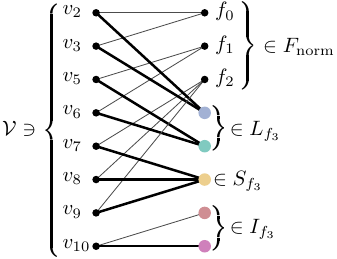}
            \caption{}
            \label{fig:generalized_factor_instance_aux_graph}
        \end{subfigure}
        \hfill
        \begin{subfigure}[tb]{0.3\textwidth}
            \centering
            \includegraphics[page=4]{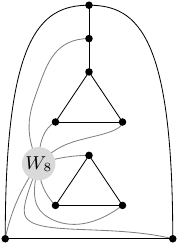}
            \caption{}
            \label{fig:generalized_factor_instance_augmentation}
        \end{subfigure}
        
        \caption{
            \subref{fig:generalized_factor_instance_graph} A planar subcubic graph~$G$.
            \subref{fig:generalized_factor_instance_aux_graph} Its corresponding \GeneralizedFactor instance.
            Thick edges denote a possible solution.
            \subref{fig:generalized_factor_instance_augmentation} A $2$-connected $3$-augmentation of~$G$ (with an indicated wheel-extension).
        }
        \label{fig:generalized_factor_instance}
    \end{figure}

    Lastly, we need to assign a set~$B(x) \subseteq \{0, 1, \ldots, \deg_A(x)\}$ of possible degrees to each vertex~$x \in \calV \cup \calF$:
    \[
        B(x) \coloneqq
        \begin{cases}
            \{1\}, & \text{if $x \in \calV$} \\
            \{0, 2, 3, \ldots, \deg_A(x)\}, & \text{if $x \in F_{\normal}$} \\
            \{0, 1, 2, \ldots, \deg_A(x)\}, & \text{if $x \in I_f$ for some connecting face~$f \in F_{\conn}$} \\
            \{1, 2, 3, \ldots, \deg_A(x)\}, & \text{if $x \in L_f$ for some connecting face~$f \in F_{\conn}$} \\
            \{2, 3, 4, \ldots, \deg_A(x)\}, & \text{if $x \in S_f$ for some connecting face~$f \in F_{\conn}$} \\
        \end{cases}
    \]

    By the following claim, the above reduction is linear.

    \begin{claim}
        \label{claim:size_of_A_is_linear}
        The order and size of~$A$ is linear in~$n$.
        Moreover, $A$ can be computed in linear time.
    \end{claim}

    \begin{claimproof}
        If $v \in V(G)$ is a vertex incident to a face~$f$ of~$\calE$, then it lies in at most three blocks of~$G_f$, since its degree in~$G_f$ is at most~$3$.
        Further, every vertex is incident to at most three different faces of~$\calE$. 
        Thus, there are at most nine blocks in~$B_f$ containing~$v$, which shows that~$\abs{B_f}$ is linear in~$n$. 
        As the number of faces of a planar embedding is linear in~$n$, so is~$\abs{A}$. 
        A vertex~$v \in \calV$ is incident to (at most) two faces of~$\calE$, and therefore it is contained in at most six distinct blocks in~$B_f$ for some faces~$f$. 
        Hence, we see that each vertex~$x \in \calV$ is adjacent to at most six vertices in~$B_f \subseteq \calF$ and at most two vertices in~$F_{\normal} \subseteq \calF$. 
        Thus, the bipartite graph~$A$ contains at most $8n$~edges. 
        Note that, in particular,~$A$ can be computed in linear time.
    \end{claimproof}

    The next two claims establish that~$A$ admits a $B$-factor if and only if~$G$ admits a $2$-connected $3$-augmentation~$H$ extending~$\calE$.

    \begin{claim}
        \label{claim:Bfactor_to_3aug}
        If~$A$ admits a $B$-factor, then~$G$ has a $2$-connected $3$-augmentation~$H$ extending~$\calE$.
    \end{claim}
    
    \begin{claimproof}
        Let~$A'$ be a $B$-factor of~$A$, i.e., a subgraph~$A'$ such that $\deg_{A'}(x) \in B(x)$ for every~$x \in V(A)$.
        We shall construct a connected and bridgeless supergraph~$H'$ of~$G$.
        (A connected graph with maximum degree~$3$ is~$2$-connected if and only if it is bridgeless.)
        We construct $H'$ as follows:
        For each edge~$vx \in E(A')$ with~$v \in \calV$ and $x \in \calF$, we add a new half-edge from~$v$ into a face~$f$ of~$\calE$.
        If~$x$ represents a face, then~$f = x$.
        Otherwise, let~$f$ be the face such that $x$ represents a block in~$B_f$.
    
        Now, for each face~$f$ of~$\calE$, all half-edges ending inside~$f$ are connected to a new vertex~$v_f$.
        Obviously, $H'$ is planar.
        To see that~$H'$ is connected, consider a connecting face~$f$ of~$\calE$.
        We have~$\deg_{A'}(b) \geq 1$ for every~$b \in S_f \cup L_f$, so each such~$b$ is connected to~$v_f$ by at least one edge.
        Lastly, to prove that~$H'$ is bridgeless, we consider three cases:
        \begin{itemize}
            \item A non-bridge of~$G$ is a non-bridge in~$H'$ as~$H' \supseteq G$.
            \item A bridge~$e$ of~$G$ has a unique face~$f$ of~$\calE$ incident to both its sides.
            The leaf blocks in~$L_f$ are subgraphs of the blocks separated by~$e$.
            As we have $\deg_{A'}(b) \geq 1$ for all~$b \in L_f$, there is at least one edge from each leaf block to~$v_f$ in~$H'$.
            Thus,~$e$ is a non-bridge in~$H'$.
            \item No edge incident to a new vertex~$v_f$ (for some face~$f$ of~$\calE$) is a bridge, because~$v_f$ has at least two edges to every incident component of~$H'$:
            If~$f$ is a normal face, then $\deg_{H'}(v_f) \geq 2$ because~$\deg_{A'}(f) \neq 1$.
            Now assume that~$f$ is a connecting face, and consider a component~$C$ of~$G_f$.
            If~$C$ consists of a single block~$b$ (which would be in~$S_f$), then~$v_f$ is connected to at least two vertices of~$b$, because we have~$\deg_{A'}(b) \geq 2$.
            Otherwise, if~$C$ consists of multiple blocks, then its block-cut-tree has at least two leaves.
            In this case,~$v_f$ is connected to at least one vertex per leaf block~$b \in L_f$ because~$\deg_{A'}(b) \geq 1$.
        \end{itemize}
        Since~$\deg_{A'}(v) = 1$ for each~$v \in \calV$, we see that all vertices in~$V(G)$ have degree~$3$ in $H'$.
        We apply a wheel-extension (\cref{obs:wheel_extension}) at each new vertex~$v_f$ of degree larger than~$3$, and replace each vertex~$v_f$ of degree~$2$ by the gadget represented in \cref{fig:preprocessing_parallel} (this simulates replacing it with a single edge connecting its neighbors, but without the risk of creating a multi-edge). 
        We obtain a $3$-regular graph~$H$ that is planar, connected, and bridgeless.
        Finally, $H$ is $2$-connected, because it is connected, bridgeless and of maximum degree~$3$.
    \end{claimproof}
    
    \begin{claim}
        \label{claim:3aug_to_Bfactor}
        If~$G$ has a $2$-connected $3$-augmentation~$H$ extending~$\calE$, then~$A$ has a~$B$-factor.
    \end{claim}

    \begin{claimproof}
        Since~$H$ extends~$\calE$, its new vertices and edges must have been added solely into the faces of~$\calE$.

        We have to construct a~$B$-factor~$A'$ of~$A$.
        To this end, we consider the vertices~$v \in \calV$, i.e., vertices of~$G$ with~$\deg_G(v) = 2$.
        Each such~$v$ has~$\deg_H(v) = 3$, so there is exactly one new edge~$e$ incident to~$v$.
        Let~$f$ be the face of~$\calE$ that~$e$ is inside.
        If~$f$ is a normal face, we add the edge~$vf$ to~$A'$.
        Now assume that~$f$ is a connecting face.
        As~$\deg_G(v) = 2$, vertex~$v$ can be in at most two blocks of~$G_f$.
        If~$v$ is in a singleton block, then this is the only block in~$B_f$ containing~$v$.
        If~$v$ lies in exactly two leaf blocks~$b_1, b_2 \in B_f$, then both must be bridges, whose other endpoints have degree~$1$; a contradiction to~$\delta(G) \geq 2$.
        Thus, there is at most one block in~$L_f \cup S_f$ containing~$v$.
        If there exists a block in~$L_f \cup S_f$ containing~$v$, let~$b$ be this (unique) block.
        Otherwise, choose an arbitrary~$b \in B_f$ containing~$v$.
        We add the edge~$vb$ to~$A'$.

        \medskip
        \noindent
        We prove that~$\deg_{A'}(x) \in B(x)$ for all~$x \in \calV \cup \calF$.
        For each vertex~$v \in \calV$, we added exactly one edge to~$A'$, therefore we have~$\deg_{A'}(v) = 1$ as required.

        For a normal face~$f \in F_{\normal}$, it holds that $\deg_{a'}(f)$ is either~$0$ or at least~$2$ because~$H$ is $2$-connected and therefore there are either no or at least two new edges inside~$f$.

        Now, consider a connecting face~$f \in F_{\conn}$.
        Each~$b \in S_f$ is a singleton block of~$G_f$.
        Since~$H$ is $2$-connected, there are (at least) two paths leaving different vertices~$v_1, v_2 \in b$ via new edges through~$f$.
        Therefore, $A'$ contains the edges~$v_1b$ and~$v_2b$, i.e., $\deg_{A'}(b) \geq 2$ as required.
        Similarly, each~$b \in L_f$ is a leaf-block.
        Since~$H$ is $2$-connected, there is (at least) one path leaving~$v \in b$ via a new edge through~$f$.
        Therefore, we have~$vb \in E(A')$ and thus~$\deg_{A'}(b) \geq 1$ as required.
    \end{claimproof}

    \medskip
    \noindent
    It remains to argue that we can compute a $B$-factor of~$A$ efficiently.
    By inspecting the set~$B(x)$ for each~$x \in \calV \cup \calF$, we can see that none of them contains a gap of size~$2$ or greater.
    Therefore, we are in a special case of the \GeneralizedFactor problem that can be solved, in~$\calO(n^4)$ time, by \Cornuejols' algorithm (see \cref{thm:generalized-factor-cornuejols}).

    A closer inspection yields that only for~$x \in S_f$ the sets~$B(x)$ contain two forbidden degrees.
    (Note that~$\deg_A(v) \leq 2$ for all~$v \in \calV$:
    If there is a face~$f$ such that~$v$ is contained in two blocks of~$G_f$, then both edges incident to~$v$ are bridges; thus~$v$ is incident to no other face. Otherwise, this follows from $\deg_G(v) \leq 2$, i.e., $v$ being incident to at most two faces.)
    Therefore, if~$S_f = \emptyset$ for every connecting face~$f \in F_{\conn}$, then we can even apply the algorithm by \Sebo, which takes only~$\calO(n^2)$ time (see \cref{thm:generalized-factor-sebo}).
    In particular, this is the case if~$G$ is connected.
\end{proof}

\begin{remark}
    The attentive reader might be tempted to think that we can modify the \textsc{Generalized Factor} instance in the proof above so that it satisfies the conditions of \Sebo, even when $S_f \neq \emptyset$.
    One such attempt resides in splitting each vertex~$x$ representing a singleton block into two vertices~$x_1, x_2$ with associated sets~$B(x_1), B(x_2)$ which only exclude the value~$0$.
    Indeed, the sets $B(x_1), B(x_2)$ fulfill the condition of \Sebo, but now some of the vertices in~$\calV$ might not, see Figure~\ref{fig:generalized_factor_splitting_vertex}.
     \begin{figure}[tb]
        \centering
        \begin{subfigure}{0.4\textwidth}
            \centering
            \includegraphics[page=2]{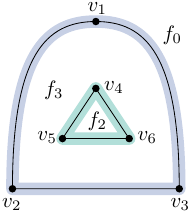}
            \caption{}
            \label{fig:generalized_factor_splitting_vertex_graph}
        \end{subfigure}
        \hspace{4mm}
        \begin{subfigure}{0.4\textwidth}
            \centering
            \includegraphics[page=3]{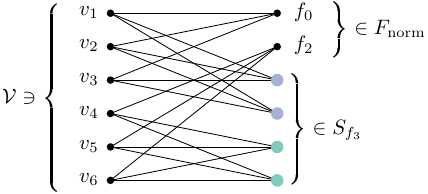}
            \caption{}
        \label{fig:generalized_factor_splitting_vertex_aux_graph}
        \end{subfigure}
        \caption{
            \subref{fig:generalized_factor_splitting_vertex_graph} A planar subcubic graph~$G$.
            \subref{fig:generalized_factor_splitting_vertex_aux_graph} Its modified \GeneralizedFactor instance where each singleton block appears twice.
        }
        \label{fig:generalized_factor_splitting_vertex}
    \end{figure}
    In the obtained \textsc{Generalized Factor} instance, all vertices in~$\calV$ have degree~$3$, yet $B(x) = \{1\}$ for every $x \in \calV$, i.e both~$2$ and $3$ are forbidden values. 
    Thus, the conditions of \Sebo (see \cref{thm:generalized-factor-sebo}) are not satisfied.
\end{remark}

\section{Connected and \texorpdfstring{$\bm{2}$}{2}-Connected \texorpdfstring{$\bm{3}$}{3}-Augmentations for Variable Embeddings}
\label{sec:1con_2con_3aug_variable_embedding}

Here we consider the $3$-augmentation problem for an arbitrary input graph~$G$ and a connected or $2$-connected output graph~$H$ in the variable embedding setting, corresponding to the second and third column of the table in \cref{fig:overview} respectively. 
We present a linear-time algorithm for connected output and a quadratic-time algorithm for $2$-connected output, which yields Statement~\ref{itm:main_theorem_1con_variable} and~\ref{itm:main_theorem_2con_variable} of \cref{thm:main_theorem}. 

In fact, we are mainly concerned with finding $2$-connected $3$-augmentations, i.e., Statement~\ref{itm:main_theorem_2con_variable}.
Our task boils down to finding a suitable planar embedding of~$G$ such that for each vertex~$v$ of~$G$ and each missing edge at~$v$, we can assign an incident face at~$v$ that should contain the new edge.
Let us note that this might only work for some planar embeddings of~$G$.
See \cref{fig:example_wrong_embedding} for a negative example.

\begin{figure}[tb]
    \centering
    \begin{subfigure}[t]{0.4\textwidth}
        \centering
        \includegraphics[page=1]{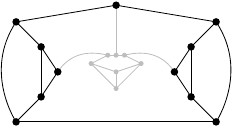}
        \caption{A planar graph~$G_1$ with a $2$-connected $3$-augmentation~$H_1$.}
        \label{fig:example_3_augmentation}
    \end{subfigure}
    \hspace{4mm}
    \begin{subfigure}[t]{0.4\textwidth}
        \centering
        \includegraphics[page=2]{figures_journal/small-illustrations.pdf}
        \caption{A different planar embedding of~$G_1$ that does not allow for a $2$-connected $3$\nobreakdash-augmentation.}
        \label{fig:example_wrong_embedding}
    \end{subfigure}
    \caption{
        Example instances for the $2$-connected $3$-augmentation problem.
    }
    \label{fig:small-illustrations}
\end{figure}

We show Statement~\ref{itm:main_theorem_2con_variable} of \cref{thm:main_theorem} in three steps.
First, we show that for $G$ to admit a $2$-connected $3$-augmentation we may restrain ourselves to all inclusion-maximal $2$-connected components, called a \emph{blocks}, of~$G$, and checking whether each admits a $2$-connected $3$-augmentation.
As all blocks can be found in linear time~\cite{Tarjan1972_Blocks}, we may assume henceforth that the input graph~$G$ is $2$-connected.
Second, we generalize \cref{thm:2con_fixed} to multi-graphs, i.e., we present a polynomial-time algorithm that given a planar multi-graph~$G$ with a fixed planar embedding~$\calE$ tests whether~$G$ admits a $2$-connected $3$-augmentation $H \supseteq G$ with a planar embedding whose restriction to~$G$ equals~$\calE$.
Finally, for a $2$-connected graph~$G$ with variable embedding, we use an SPQR-tree of~$G$ to efficiently go through the possible planar embeddings of~$G$ with a dynamic program and to identify one such embedding that allows for a $2$-connected $3$-augmentation, or conclude that no such exists.

The next proposition settles the $3$-augmentation problem for connected output, and yields a reduction to $2$-connected input for the decision problem with $2$-connected output.

\begin{proposition}\label{prop:wlog-connected}
    For a disconnected subcubic graph~$G$ with connected components $G_1,\ldots,G_{\ell}$, $\ell \geq 2$, we have that
    \begin{enumerate}[(i)]
        \item $G$ has a connected $3$-augmentation if and only if no~$G_i$ is $3$-regular
        \item $G$ has a $2$-connected $3$-augmentation if and only if each~$G_i$ has a $2$-connected $3$-augmentation and no~$G_i$ is $3$-regular.
    \end{enumerate}
\end{proposition}

\begin{proof}
    Let $G_1,\ldots,G_{\ell}$ denote the components of~$G$.
    \begin{enumerate}[(i)]
        \item Suppose there exists a connected $3$-augmentation~$H$ of~$G$. Note that~$H$ is also a connected $3$-augmentation of each~$G_i$.
        As~$H$ is connected, each~$G_i$ has a vertex incident to at least one edge in~$E(H)-E(G)$, i.e., no~$G_i$ is $3$-regular.

        Suppose now that each~$G_i$ admits a connected $3$-augmentation~$H_i$ and no~$G_i$ is $3$-regular.
        Each~$H_i$ contains a new edge~$e_i$. 
        Consider an embedding of $H_1 \dotcup \cdots \dotcup H_{\ell}$ where each~$e_i$ lies on the outer face.
        Finally, add a copy of $K_4^{(2\ell)}$ into the outer face, delete $e_1,\ldots,e_{\ell}$, and connect the $2\ell$ degree-$2$ vertices of $H_1 \dotcup \cdots \dotcup H_{\ell}$ with the $2\ell$ degree-$2$ vertices of $K_4^{(2\ell)}$ by a non-crossing matching.
        The resulting graph is a connected $3$-augmentation of~$G$. 

        \item Again, any $2$-connected $3$-augmentation of~$G$ is in particular a $2$-connected $3$-augmentation of each~$G_i$. 
        If such an augmentation of~$G$ exists, no~$G_i$ is $3$-regular. 
        
        On the other hand, if each~$G_i$ admits a $2$-connected $3$-augmentation~$H_i$ and no~$G_i$ is $3$-regular, the same construction as in the proof above yields a connected $3$-augmentation~$H$ of~$G$.
        Note that~$H$ is bridgeless.
        As a graph of maximum degree~$3$ is $2$-connected if it is connected and bridgeless, the claim follows.
        \qedhere
    \end{enumerate}
\end{proof}

The first claim of \cref{prop:wlog-connected} yields a linear-time algorithm for arbitrary input and connected output in the variable embedding setting.
This corresponds to Statement~\ref{itm:main_theorem_1con_variable} of \cref{thm:main_theorem}.

The remainder of \cref{sec:1con_2con_3aug_variable_embedding} only considers $2$-connected output.
The second claim of above proposition shows that we may assume that the input is connected.
By the following proposition, we may even assume that it is $2$-connected.

\begin{proposition}
    \label{prop:wlog-2-connected}
    A connected subcubic graph~$G$ admits a $2$-connected $3$-augmentation if and only if each block of~$G$ admits a $2$-connected $3$-augmentation.
\end{proposition}

\begin{proof}
    Suppose~$G$ contains no bridge. 
    As~$G$ is connected and subcubic, it follows that $G$ is $2$-connected, i.e. the only block of~$G$ corresponds to the whole graph and the claim follows.
    
    We may therefore assume that~$G$ contains a bridge $e=uv$. 
    Let~$G_1$ be the connected component of $G - e$ containing~$u$, and let $G_2 = G - G_1$ denote the remaining graph.
    It is enough to show that if~$G_1$ and~$G_2$ have $2$-connected $3$-augmentations~$H_1$ respectively~$H_2$, then~$G$ also has a $2$-connected $3$-augmentation.
    To this end, consider an edge $e_1 \in E(H_1)-E(G_1)$ incident to~$u$ and an edge $e_2 \in E(H_2)-E(G_2)$ incident to~$v$.
    These edges exist as $\deg_{G_1}(u), \deg_{G_2}(v) \leq \Delta(G)-1 \leq 2$ but $\deg_{H_1}(u) = \deg_{H_2}(v) = 3$.
    Choose a planar embedding of $H_1 \dotcup H_2$ with~$e_1$ and~$e_2$ being outer edges.
    Denoting by $a,b$ the endpoints of $e_1,e_2$ different from $u,v$, we see that $(H_1 - e_1) \dotcup (H_2 - e_2) \cup \{uv, ab\}$ is a connected $3$-augmentation of~$G$ with no bridges, i.e., a $2$-connected $3$-augmentation.
\end{proof}

\subsection{The Fixed Embedding Setting for Multigraphs}
\label{subsec:2con-fixed-embedding-multi-graph}

In order to settle the variable embedding setting for $2$-connected output, we use the algorithm for fixed embedding as a crucial subroutine.  
Yet, the embedded graphs we consider might be not simple.
The following lemma yields a linear-time reduction from multi-graphs to simple graphs. 

\begin{lemma}
    \label{prop:wlog-no-multigraph}
    Let~$G$ be an $n$-vertex planar multi-graph of maximum degree $\Delta(G) \leq 3$ with an embedding~$\calE$.
    There is a simple planar graph~$G'$ with $\Delta(G') \leq 3$ on at most $4n$~vertices with an embedding~$\calE'$ such that~$G$ has a $2$-connected $3$-augmentation extending~$\calE$ if and only if~$G'$ has one extending~$\calE'$.
\end{lemma}
\begin{proof}
    Replacing multi-edges of~$G$ with the gadget represented in \cref{fig:preprocessing_parallel} yields a simple planar graph~$G'$ of maximum degree at most~$3$ whose embedding~$\calE'$ is closely related to~$\calE$.
    As every vertex of~$G$ is incident to at most three edges, we have $\abs{E(G)} \leq 3n$.
    We introduced at most four new vertices for each edge of~$G$, thus $\abs{V(G')} \leq 4n$.
    Note that a $2$-connected $3$-augmentation of~$G$ extending~$\calE$ yields one of~$G'$ extending $\calE$ and vice versa. 
\end{proof}

We can therefore generalize \cref{thm:2con_fixed} to multi-graphs.
\begin{corollary}
    \label{prop:2con-3aug-fixed-multi-graph}
    Let~$G$ be an $n$-vertex planar multi-graph of maximum degree $\Delta(G) \leq 3$ with a fixed planar embedding~$\calE$.
    Then we can compute in time $\mathcal{O}(n^2)$ a $2$-connected $3$\nobreakdash-augmentation~$H$ of~$G$ with a planar embedding~$\calE_H$ whose restriction to~$G$ equals~$\calE$, or conclude that no such exists.
\end{corollary}

\subsection{\texorpdfstring{$\bm{2}$}{2}-Connected \texorpdfstring{$\bm{3}$}{3}-Augmentations for \texorpdfstring{$\bm{2}$}{2}-Connected Input and Variable Embeddings}
\label{subsec:variable-embedding}

Even an unlabeled $2$-connected subcubic planar graph~$G$ can have exponentially many different planar embeddings (e.g., the $(2 \times n)$-grid graph).
Thus, iterating over all embeddings of~$G$ and applying the algorithm from \cref{thm:2con_fixed} to each of them is not a polynomial-time algorithm and hence no feasible approach for us.
In this section we describe how to use the SPQR-tree of~$G$ to efficiently find a planar embedding~$\calE$ of~$G$ such that there is a $3$-augmentation~$H$ of~$G$ extending~$\calE$, or conclude that no such embedding exists.
The algorithm from \cref{prop:2con-3aug-fixed-multi-graph} will be an important subroutine. We show the following.

\begin{proposition}
    \label{prop:variable-embedding}
    Let~$G$ be an~$n$-vertex $2$-connected subcubic planar graph.
    Then we can compute in $\mathcal{O}(n^2)$ time a $2$-connected $3$-augmentation~$H$ of~$G$ or conclude that no such exists.
\end{proposition}

Together with \cref{prop:wlog-connected} and \cref{prop:wlog-2-connected}, this yields Statement~\ref{itm:main_theorem_2con_variable} of \cref{thm:main_theorem}.

\subparagraph{Overview.}
The proof of \cref{prop:variable-embedding} uses a bottom-up dynamic programming approach on the SPQR-tree~$T$ of~$G$ rooted at a Q-vertex~$\rho$ corresponding to some edge~$e_\rho$ in~$G$.
Consider a vertex~$\mu \neq \rho$ in~$T$.
Let~$uv$ be the virtual edge in~$\skel(\mu)$ that is associated to the parent edge of~$\mu$.
Recall that each embedding~$\calE$ of~$G$ with~$e_\rho$ on the outer face, when restricted to the pertinent graph~$\pert(\mu)$, gives an embedding~$\calE_{\mu}$ of~$\pert(\mu)$ whose inner faces are also inner faces of~$\calE$, and with~$u$ and~$v$ being outer vertices of~$\calE_\mu$.
The outer face of~$\calE_\mu$ is composed of two (not necessarily edge-disjoint) $u$-$v$-paths; the left and right outer path of~$\calE_\mu$, which are contained in the left and right outer face of~$\calE_\mu$ inside~$\calE$, respectively.
We seek to partition the (possibly exponentially many) planar embeddings of~$\pert(\mu)$ with~$u,v$ on its outer face into a constant number of equivalence classes based on how many edges in a $2$-connected $3$-augmentation of~$G$ could possibly ``connect'' $\pert(\mu)$ with the rest of the graph~$G$ inside the left or right outer face of $\calE_\mu$ inside~$\calE$.
This corresponds\footnote{
    up to the fact that left and right outer path may share degree-$2$ vertices, each of which sends however its third edge into only one of the left and right outer face
} to the number of degree-$2$ vertices on the left and right side in so-called \emph{inner augmentations} of~$\calE_\mu$.
Loosely speaking, it will be enough for us to distinguish three cases for the left side ($0$,~$1$, or at least~$2$ connections), the symmetric three cases for the right side, and to record which of the nine resulting combinations are possible.
Note that this grouping of embeddings of~$\calE_\mu$ into constantly many classes is the key insight that allows an efficient dynamic program.

Whether a particular equivalence class is realizable by some planar embedding~$\calE_\mu$ of~$\pert(\mu)$ will depend on the vertex type of~$\mu$ (S-, P- or R-vertex) and the realizable equivalence classes of its children~$\mu_1, \ldots, \mu_k$.
In the end, we shall conclude that the whole graph~$G$ has a $2$-connected $3$-augmentation if and only if for the unique child~$\mu$ of the root~$\rho$ of $T$ the equivalence class of embeddings of~$\pert(\mu)$ for which neither the left nor the right side has any connections is non-empty.

Most of our arguments are independent of SPQR-trees and we instead consider so-called $uv$-graphs, which are slightly more general than pertinent graphs.
We shall introduce inner augmentations of $uv$-graphs, which then give rise to label sets for $uv$-graphs, both in a fixed and variable embedding setting.
These label sets encode the aforementioned number of connections between the $uv$-graph as a subgraph of~$G$ and the rest of~$G$ in a potential $2$-connected $3$-augmentation.
After showing that we can compute variable label sets by resorting to the fixed embedding case and \cref{prop:2con-3aug-fixed-multi-graph}, we present the final dynamic program along the rooted SPQR-tree~$T$ of~$G$.

\subparagraph{$uv$-Graphs and Labels.}
A \emph{$uv$-graph} is a connected multigraph~$G_{uv}$ with $\Delta(G_{uv}) \leq 3$, two distinguished vertices~$u,v$ of degree at most~$2$, together with a planar embedding~$\calE_{uv}$ such that~$u$ and $v$ are outer vertices.
A connected multigraph~$H_{uv} \supseteq G_{uv}$ with planar embedding~$\calE_H$ is an \emph{inner augmentation} of~$G_{uv}$ if
\begin{itemize}
    \item $\calE_H$ extends~$\calE_{uv}$ and has~$u,v$ on its outer face,
    \item each of $u,v$ has the same degree in $H_{uv}$ as in $G_{uv}$,
    \item every vertex of~$H_{uv}$ except for~$u,v$ has degree~$1$ or~$3$,
    \item every degree-$1$ vertex of~$H_{uv}$ lies in the outer face of~$\calE_H$ and
    \item every bridge of~$H_{uv}$ that is not a bridge of~$G_{uv}$ is incident to a degree-$1$ vertex.
\end{itemize}
Because~$u,v$ are outer vertices in~$\calE_H$, one could add another edge~$e_{uv}$ (oriented from~$u$ to~$v$) into the outer face of~$\calE_H$ preserving planarity (this edge is not part of the inner augmentation).
Then~$e_{uv}$ splits the outer face into two faces~$f_A,f_B$ left and right of~$e_{uv}$, respectively.
Each degree-$1$ vertex of~$H_{uv}$ now lies either inside~$f_A$ or~$f_B$.

We are interested in the number of degree-$1$ vertices in each of these faces of~$\calE_H$ and write $d(H_{uv},\calE_H) = (a,b)$ if an inner augmentation~$H_{uv}$ of~$G_{uv}$ has exactly~$a$ degree-$1$ vertices inside~$f_A$ and exactly~$b$ degree-$1$ vertices inside~$f_B$.

\begin{lemma}
    \label{lem:wlog-0-or-1}
    Let~$H_{uv}$ be an inner augmentation of~$G_{uv}$ with $d(H_{uv}, \calE_H) = (a,b)$.
    If $a \geq 2$, then $G_{uv}$ has an inner augmentation $H_{uv}^0$ with $d(H_{uv}^0, \calE_H^0) = (0,b)$ and an inner augmentation $H_{uv}^1$ with $d(H_{uv}^1,  \calE_H^1) = (1,b)$.
    A symmetric statement holds when $b \geq 2$.
\end{lemma}

\begin{proof}
    Add the edge~$uv$ to the inner augmentation~$H_{uv}$ such that there are $a$~degree-$1$ vertices in~$f_A$.
    We add a copy of~$K_4^{(a)}$ into~$f_A$ and identify the~$a$ degree-$2$ vertices of~$K_4^{(a)}$ with the~$a$ degree-$1$ vertices in~$f_A$ in a non-crossing way.
    Ignoring edge~$uv$, the obtained graph is the desired inner augmentation~$H_{uv}^0$ with $d(H_{uv}^0, \calE_H^0) = (0,b)$.
    We obtain~$H_{uv}^1$ by additionally subdividing an edge of~$K_4^{(a)}$ that is incident to~$f_A$ once and by attaching a degree-$1$ vertex to it into~$f_A$.
\end{proof}

Motivated by \cref{lem:wlog-0-or-1}, we focus on inner augmentations~$H_{uv}$ with~$d(H_{uv}, \calE_H) = (a,b)$ where~$a,b \in \{0,1\}$, and assign to $H_{uv}$ in this case the \emph{label}~$\lab{ab}$ with~$\lab{a},\lab{b} \in \{\lab{0},\lab{1}\}$.

The \emph{embedded label set}~$\LSemb{G_{uv}}{\calE_{uv}}$ contains all labels~$\lab{ab}$ such that there is an inner augmentation~$H_{uv}$ of~$G_{uv}$ with label~$\lab{ab}$.
Allowing other planar embeddings of~$G_{uv}$, we further define the \emph{variable label set} as $\LSvar{G_{uv}} = \bigcup_{\calE} \LSemb{G_{uv}}{\calE}$, where~$\calE$ runs over all planar embeddings of~$G_{uv}$ where~$u$ and $v$ are outer vertices.
As this in particular includes for each embedding $\calE$ of $G_{uv}$ also the flipped embedding $\calE'$ of $G_{uv}$, it follows that $\lab{a}\lab{b} \in \LSvar{G_{uv}}$ if and only if $\lab{b}\lab{a} \in \LSvar{G_{uv}}$.
Whenever this property holds for a (variable or embedded) label set, we call the label set \emph{symmetric}.
Hence, all variable label sets are symmetric, but embedded label sets may or may not be symmetric.

For brevity, let us use~$\star$ as a wildcard character, in the sense that if $\{\lab{x0}, \lab{x1}\}$ is in an embedded or variable label set for some $\lab{x} \in \{\lab{0},\lab{1}\}$, then we shorten the notation and replace them by a label $\lab{x}\star$.
Symmetrically, we use the notation $\star\lab{x}$ and in particular define $\{\star\star\} \coloneqq \{\lab{00}, \lab{01}, \lab{10},\lab{11}\}$.
Using this notation, the eight possible symmetric label sets are:
\begin{equation}
    \label{eqn:variable-label-sets}
    \emptyset,
    \{\lab{00}\},
    \{\lab{01}, \lab{10}\},
    \{\lab{11}\},
    \{\lab{0}\star, \star\lab{0}\},
    \{\lab{00}, \lab{11}\},
    \{\lab{1}\star, \star\lab{1}\},
    \{\star\star\}
\end{equation}

The following lemma reveals the significance of inner augmentations and label sets.

\begin{lemma}
    \label{lem:3aug-if-00-label}
    Let~$G$ be a $2$-connected graph with $\Delta(G) \leq 3$ and~$\calE$ be an embedding of~$G$ with some outer edge~$e = xy$.
    Further, let~$G_{uv}$ be the $uv$-graph obtained from~$G$ by deleting~$e$ and adding two new vertices $u,v$ with edges~$ux$ and~$vy$ into the outer face of~$\calE$.
    Then~$G$ has a $2$-connected $3$-augmentation if and only if $\lab{0}\lab{0} \in \LSvar{G_{uv}}$.
\end{lemma}

\begin{proof}
    First let $H \supseteq G$ be a $2$-connected $3$-augmentation of~$G$ and let~$\calE_H$ be an embedding of~$H$ with $e=xy$ being an outer edge.
    Then deleting~$e$ and adding two new vertices $u,v$ with edges~$ux$ and~$vy$ into the outer face of~$\calE_H$ results in an inner augmentation~$H_{uv}$ of~$G_{uv}$ with respect to the embedding of~$G_{uv}$ inherited from~$\calE_H$.
    As adding an edge~$e_{uv}$ from~$u$ to~$v$ into~$H_{uv}$ gives a graph with no degree-$1$ vertices, we have $\lab{0}\lab{0} \in \LSvar{G_{uv}}$.
    
    Conversely, assume that $\lab{0}\lab{0} \in \LSvar{G_{uv}}$.
    Then there is an embedding~$\calE_{uv}$ of~$G_{uv}$ that allows for some inner augmentation~$H_{uv}$ with embedding~$\calE_H$ for which $H_{uv} + e_{uv}$ has no degree-$1$ vertices, where $e_{uv} = uv$ denotes a new edge between~$u$ and~$v$.
    Thus, in $H_{uv}$ the vertices~$u$ and~$v$ have degree~$1$ (as in~$G_{uv}$), every vertex of~$H_{uv}$ except~$u,v$ has degree~$3$, and the only bridges of~$H_{uv}$ are the edges~$ux$ and~$vy$.
    Then we obtain a $2$-connected $3$-augmentation~$H$ of~$G$ by removing $ux,vy$ from~$H_{uv}$ and adding the edge~$xy$ into the outer face of~$\calE_H$.
    In case, $H_{uv}$ already contains the edge $xy$, this is replaced by a copy of $K_4^{(2)}$ with two non-crossing edges between $x,y$ and the two degree-$2$ vertices of $K_4^{(2)}$.
\end{proof}

\subparagraph{Gadgets.}
In our algorithm below, we aim to replace certain $uv$-graphs~$X$ (with variable embedding) by $uv$-graphs~$Y$ with fixed embedding~$\calE_Y$, such that the variable label set~$\LSvar{X}$ equals the embedded label set~$\LSemb{Y}{\calE_Y}$.
This will allow us to use \cref{prop:2con-3aug-fixed-multi-graph} from the fixed embedding setting as a subroutine.

The following \lcnamecref{lem:gadgets} describes seven $uv$-graphs, each with a fixed embedding, corresponding to the seven different non-empty variable label sets as given in~\eqref{eqn:variable-label-sets}.
For this purpose, each such \emph{gadget} is itself a $uv$-graph~$Y$ with a fixed embedding~$\calE_Y$.

\begin{lemma}
    \label{lem:gadgets}
    For every $uv$-graph~$G_{uv}$ with $\LSvar{G_{uv}} \neq \emptyset$ there exists a gadget~$Y$ with an embedding~$\calE_Y$ such that~$u,v$ are outer vertices and~$\LSemb{Y}{\calE_Y} = \LSvar{G_{uv}}$.
\end{lemma}

\begin{figure}[bt]
    \centering
    \begin{subfigure}[b]{0.23\textwidth}
        \centering
        \includegraphics[page=1]{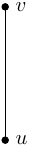}
        \caption{$\LSemb{Y}{\calE_Y} \!=\! \{\lab{00}\}$}
        \label{fig:gadget_00}
    \end{subfigure}
    \begin{subfigure}[b]{0.26\textwidth}
        \centering
        \includegraphics[page=2]{figures_journal/gadgets.pdf}
        \caption{$\LSemb{Y}{\calE_Y} \!=\! \{\lab{01}, \lab{10}\}$}
        \label{fig:gadget_01}
    \end{subfigure}
    \begin{subfigure}[b]{0.26\textwidth}
        \centering
        \includegraphics[page=6]{figures_journal/gadgets.pdf}
        \caption{$\LSemb{Y}{\calE_Y} \!=\! \{\lab{00},\lab{11}\}$}
        \label{fig:gadget_0011}
    \end{subfigure}
    \begin{subfigure}[b]{0.22\textwidth}
        \centering
        \includegraphics[page=4]{figures_journal/gadgets.pdf}
        \caption{$\LSemb{Y}{\calE_Y} \!=\! \{\lab{11}\}$}
        \label{fig:gadget_11}
    \end{subfigure}
    \par\bigskip
    \begin{subfigure}[b]{0.23\textwidth}
        \centering
        \includegraphics[page=3]{figures_journal/gadgets.pdf}
        \caption{$\LSemb{Y}{\calE_Y} \!=\! \{\star\star\}$}
        \label{fig:gadget_ss}
    \end{subfigure}
    \begin{subfigure}[b]{0.26\textwidth}
        \centering
        \includegraphics[page=5]{figures_journal/gadgets.pdf}
        \caption{$\LSemb{Y}{\calE_Y} \!=\! \{\lab{0}\star, \star\lab{0}\}$}
        \label{fig:gadget_0s}
    \end{subfigure}
    \begin{subfigure}[b]{0.26\textwidth}
        \centering
        \includegraphics[page=7]{figures_journal/gadgets.pdf}
        \caption{$\LSemb{Y}{\calE_Y} \!=\! \{\lab{1}\star, \star\lab{1}\}$}
        \label{fig:gadget_1s}
    \end{subfigure}
    \begin{subfigure}[b]{0.22\textwidth}
        \centering
        \includegraphics[page=8]{figures_journal/gadgets.pdf}
        \caption{Only degree $3$.}
        \label{fig:gadget_unsaturated}
    \end{subfigure}
    \caption{
        \subref{fig:gadget_00}--\subref{fig:gadget_1s}~The seven gadgets for the seven non-empty variable label sets~$\LSvar{G_{uv}}$ of a $uv$-graph~$G_{uv}$.
        \subref{fig:gadget_unsaturated}~Modification to locally replace a degree-$2$ vertex by four degree-$3$-vertices.
    }
    \label{fig:gadgets}
\end{figure}

\begin{proof}
    We distinguish the seven cases of what~$\LSvar{G_{uv}}$ is according to \eqref{eqn:variable-label-sets}.
    In each of our gadgets, which are shown in \cref{fig:gadgets}, vertices~$u$ and~$v$ are the only degree~$1$ vertices.
    For convenience, let us call the degree-$2$ vertices in the gadgets~$Y$ the \emph{white vertices}.
    Hence, each white vertex~$w$ (names as in \cref{fig:gadgets}) has exactly one new incident edge in an inner augmentation~$H$ of~$Y$, and we denote this new edge by~$e_w$.
    Observe that all gadgets, except for the one in \cref{fig:gadget_ss} for label set~$\{\star\star\}$, have at most one white vertex on either side (left or right) of the outer face of~$\calE_Y$.
    Hence whenever such a white vertex~$w$ has its new edge~$e_w$ in the outer face of~$\calE_Y$, then~$e_w$ is a bridge in~$H$ and its other endpoint has degree~$1$, thus counting towards~$d(H,\calE_H)$.
    Also for convenience, we call an inner face of~$\calE_Y$ that has~$0$ or~$1$ incident white vertex a \emph{gray face}.
    Observe that gray faces contain no new edges of~$H$ as inner faces may not contain bridges of~$H$.
    
    We proceed by going through the seven gadgets one-by-one.
    Note that in each case it is enough to check which of $\lab{00}$, $\lab{01}$, $\lab{10}$ and $\lab{11}$ are contained in a label set in order to determine it uniquely.
    For this we always let~$H$ be any inner augmentation of~$Y$.
    
    \begin{itemize}
        \item Case $\LSvar{G_{uv}} = \{\lab{00}\}$, i.e., $\lab{00} \in \LSvar{G_{uv}}$ and $\lab{01},\lab{10},\lab{11} \notin \LSvar{G_{uv}}$.
        Consider the gadget~$Y$ and its embedding~$\calE_Y$ represented in Figure~\ref{fig:gadget_00}. 
        We have no white vertices and thus only one inner augmentation~$H = Y$.
        Thus, $\LSemb{Y}{\calE_Y} = \{\lab{00}\}$ follows.
        
        \item Case $\LSvar{G_{uv}} = \{\lab{01}, \lab{10}\}$, i.e., $\lab{01},\lab{10} \in \LSvar{G_{uv}}$ and $\lab{00},\lab{11} \notin \LSvar{G_{uv}}$.
        Consider the gadget~$Y$ and its embedding~$\calE_Y$ represented in Figure~\ref{fig:gadget_01}. 
        The only white vertex~$x$ has its edge~$e_x$ in~$H$ in the outer face.
        Clearly there are exactly two possibilities; $e_x$ on the left or on the right.
        Thus, we have $\LSemb{Y}{\calE_Y} = \{\lab{01},\lab{10}\}$.
        
        \item Case $\LSvar{G_{uv}} = \{\lab{00},\lab{11}\}$, i.e., $\lab{00},\lab{11} \in \LSvar{G_{uv}}$ and $\lab{01},\lab{10} \notin \LSvar{G_{uv}}$.
        Consider the gadget~$Y$ and its embedding~$\calE_Y$ represented in Figure~\ref{fig:gadget_0011}. 
        We have two white vertices~$\ell$ and~$r$.
        If one of $e_\ell, e_r$ lies in the inner face~$f$ of~$\calE_Y$, then the other also lies in~$f$, as otherwise there would be a bridge of~$H$ in~$f$. 
        Thus either both edges~$e_\ell$ and~$e_r$ are in~$f$ or none, so $\lab{01},\lab{10} \notin \LSemb{Y}{\calE_Y}$.
        To obtain $\lab{00} \in \LSemb{Y}{\calE_Y}$, add an edge $e = e_\ell = e_r$ between~$\ell$ and~$r$ in~$f$.
        To obtain $\lab{11} \in \LSemb{Y}{\calE_Y}$, put~$e_\ell$,~$e_r$ into the outer face of $\calE_Y$.
        
        \item Case $\LSvar{G_{uv}} = \{\lab{11}\}$, i.e., $\lab{11} \in \LSvar{G_{uv}}$ and $\lab{00}, \lab{01}, \lab{10} \notin \LSvar{G_{uv}}$.
        Consider the gadget~$Y$ and its embedding~$\calE_Y$ represented in Figure~\ref{fig:gadget_11}. 
        We have two white vertices~$\ell,r$.
        Since both inner faces are gray, we have that~$e_\ell$ and~$e_r$ lie on separate sides of the outer face and form bridges of~$H$.
        Thus, there exists exactly one inner augmentation of~$Y$ for this embedding~$\calE_Y$ and we have $\LSemb{Y}{\calE_Y} = \{\lab{11}\}$.
        
        \item Case $\LSvar{G_{uv}} = \{\star\star\}$, i.e., $\lab{00},\lab{01},\lab{10},\lab{11} \in \LSvar{G_{uv}}$. 
        Consider the gadget~$Y$ and its embedding~$\calE_Y$ represented in Figure~\ref{fig:gadget_ss}. 
        We have three white vertices~$x$, $y$ and~$z$.
        To obtain $\lab{00} \in \LSemb{Y}{\calE_Y}$, put a new vertex into the outer face of~$\calE_Y$ and connect it to $x$, $y$, $z$ by the edges $e_x$, $e_y$, $e_z$, respectively.
        To obtain $\lab{01} \in \LSemb{Y}{\calE_Y}$, add an edge $e = e_x = e_y$ between~$x$ and~$y$ in the outer face and put~$e_z$ into the outer face on the right.
        Symmetrically, we obtain $\lab{10} \in \LSemb{Y}{\calE_Y}$.
        To obtain $\lab{11} \in \LSemb{Y}{\calE_Y}$, put~$e_x$ into the outer face on the left, add a new vertex~$w$ into the outer face on the right, connect~$w$ to~$y$, $z$ by the edges $e_y$, $e_z$, respectively, and add a pendant edge at~$w$ into the outer face of the result.
    
        \item Case $\LSvar{G_{uv}} = \{\lab{0}\star, \star\lab{0}\}$, i.e., $\lab{0}\lab{0},\lab{0}\lab{1},\lab{1}\lab{0} \in \LSvar{G_{uv}}$ and $\lab{1}\lab{1} \notin \LSvar{G_{uv}}$.
        Consider the gadget~$Y$ and its embedding~$\calE_Y$ represented in Figure~\ref{fig:gadget_0s}. 
        We have three white vertices $\ell$, $x$, $r$ and the edge~$e_x$ in~$H$ lies in the face~$f$ of~$\calE_Y$ with all three white vertices on its boundary (since~$e_x$ may not lie in the gray face).
        Since~$e_x$ is not a bridge, at least one of~$e_\ell$ and~$e_r$ lies within~$f$ as well. 
        Therefore,~$e_\ell$ and~$e_r$ cannot both lie in the outer face of~$Y$, so $\lab{11} \notin \LSemb{Y}{\calE_Y}$ follows. 
        To obtain $\lab{00} \in \LSemb{Y}{\calE_Y}$, put a new vertex into~$f$ and connect it to $\ell$, $x$, $r$ by the edges $e_\ell$, $e_x$, $e_r$, respectively.
        To obtain $\lab{01} \in \LSemb{Y}{\calE_Y}$, put~$e_\ell$ into the outer face of~$\calE_Y$, and add an edge $e = e_x = e_y$ between~$x$ and~$y$ in~$f$.
        Symmetrically, we obtain that $\lab{01} \in \LSemb{Y}{\calE_Y}$.
    
        \item Case $\LSvar{G_{uv}} = \{\lab{1}\star, \star\lab{1}\}$, i.e., $\lab{01},\lab{10},\lab{11} \in \LSvar{G_{uv}}$ and $\lab{00} \notin \LSvar{G_{uv}}$.
        Consider the gadget~$Y$ and its embedding represented in Figure~\ref{fig:gadget_1s}. 
        We have five white vertices $\ell$, $x$, $y$, $z$, and~$r$.
        Let~$f$ be the inner face of~$\calE_Y$ with $x$, $y$ and~$z$ on its boundary.
        As the other face at~$y$ is gray, edge~$e_y$ lies in~$f$.
        Now suppose that none of $e_\ell$, $e_r$ lie in the outer face of~$\calE_Y$.
        As then~$e_\ell$ is not a bridge,~$e_x$ lies not in~$f$.
        Similarly, as~$e_r$ is then not a bridge,~$e_z$ lies not in~$f$.
        But then~$e_y$ is a bridge in the inner face~$f$, which is impossible for an inner augmentation.
        Hence $\lab{00} \notin \LSemb{Y}{\calE_Y}$.
    
        To obtain $\lab{01} \in \LSemb{Y}{\calE_Y}$, put~$e_r$ into the outer face of~$\calE_Y$, add an edge $e = e_y = e_z$ between~$y$ and~$z$ into~$f$, and an edge $e' = e_\ell = e_x$ between~$\ell$ and~$x$ into their common inner face in~$\calE_Y$.
        Symmetrically, we obtain $\lab{10} \in \LSemb{Y}{\calE_Y}$.
        Finally, to obtain $\lab{11} \in \LSemb{Y}{\calE_Y}$, put both~$e_\ell$ and~$e_r$ into the outer face of~$\calE_Y$, add a new vertex into~$f$ and connect it to $x$, $y$, $z$ by edges $e_x$, $e_y$, $e_z$, respectively.
        \qedhere
    \end{itemize}
\end{proof}

We remark that \cref{fig:gadget_unsaturated} is not a gadget, and instead a local modification that we use at other places, such as in the proof of \cref{lem:label-check}.

\subparagraph{Computing a Label Set.}
In our algorithm below we want to compute the variable label sets of~$\pert(\mu)$ for vertices~$\mu$ of the rooted SPQR-tree $T$ of $G$.
As we will see, we can reduce this to a constant number of computations of embedded label sets of certain $uv$-graphs that are specifically crafted to encode all the possible embeddings of~$\pert(\mu)$.
The following \lcnamecref{lem:label-check} describes how to do this.

\begin{lemma}
    \label{lem:label-check}
    Let~$G_{uv}$ be an~$n$-vertex $uv$-graph and~$\calE_{uv}$ a planar embedding where~$u$ and~$v$ are outer vertices.
    Then we can check each of the following in time~$\mathcal{O}(n^2)$:
    \begin{itemize}
        \item Whether~$\lab{00} \in \LSemb{G_{uv}}{\calE_{uv}}$.
        \item Whether~$\lab{01} \in \LSemb{G_{uv}}{\calE_{uv}}$ or~$\lab{10} \in \LSemb{G_{uv}}{\calE_{uv}}$.
        \item Whether~$\lab{11} \in \LSemb{G_{uv}}{\calE_{uv}}$.
    \end{itemize}
    In particular, if~$\LSemb{G_{uv}}{\calE_{uv}}$ is symmetric, then this is sufficient to determine the exact embedded label set~$\LSemb{G_{uv}}{\calE_{uv}}$.
\end{lemma}

\begin{proof}
    The eight possible symmetric label sets in~\eqref{eqn:variable-label-sets} correspond bijectively to the eight possible yes-/no-answer combinations of the above three checks.
    Hence these checks are sufficient to determine~$\LSemb{G_{uv}}{\calE_{uv}}$, provided it is symmetric.
    \begin{itemize}
        \item To check whether~$\lab{00} \in \LSemb{G_{uv}}{\calE_{uv}}$, add an edge~$e_{uv}$ between~$u,v$ into the outer face of~$\calE_{uv}$.
        If afterwards~$u$ and/or~$v$ has degree~$2$, then we further replace them with degree\nobreakdash-$3$ vertices using the gadget shown in \cref{fig:gadget_unsaturated}.
        We call the obtained embedded planar multigraph~$G_{uv}^+$.
        We claim that $\lab{00} \in \LSemb{G_{uv}}{\calE_{uv}}$ if and only if~$G_{uv}^+$ has a $2$-connected $3$-augmentation extending its planar embedding.
        
        Indeed, removing edge~$e_{uv}$ from any such $3$-augmentation (and possibly undoing the replacements of~$u,v$) yields an inner augmentation~$H_{uv}$ of~$G_{uv}$ with label~$\lab{00}$.
        
        On the other hand, an inner augmentation~$H_{uv}$ with label~$\lab{00}$ that extends~$\calE_{uv}$ and an additional edge between~$u,v$ is a $2$-connected $3$-augmentation, except that~$u,v$ might still have degree~$2$ (they have degree~$1$ or~$2$ in~$G_{uv}$ which is increased by one by~$e_{uv}$).
        In that case, replace them by the gadget shown in~\cref{fig:gadget_unsaturated} to obtain a $2$-connected $3$-augmentation of $G_{uv}^+$.
    
        \item Similarly, to check whether~$\lab{01}$ or~$\lab{10}$ is in~$\LSemb{G_{uv}}{\calE_{uv}}$, we add a path of length two between~$u,v$ into the outer face of~$\calE_{uv}$.
        Let~$w$ be the middle vertex of this path.
        If afterwards~$u$ and/or~$v$ have degree~$2$, we replace them with degree-$3$ vertices using the gadget from \cref{fig:gadget_unsaturated}.
        Call the obtained planar graph~$G_{uv}^+$.
        
        We claim that $\{\lab{01}, \lab{10}\} \cap \LSemb{G_{uv}}{\calE_{uv}} \neq \emptyset$ if and only if~$G_{uv}^+$ has a $2$-connected $3$-augmentation extending its planar embedding.
        
        If such a $3$-augmentation exists, then removing edges~$uw$ and~$vw$ (and possibly undoing the replacements of~$u,v$) yields an inner augmentation~$H_{uv}$ of~$G_{uv}$ where~$w$ is the only degree-$1$ vertex (and in the outer face of~$H_{uv}$).
        It follows that~$H_{uv}$ has label~$\lab{01}$ or~$\lab{10}$.
        
        For the other direction, assume that~$\lab{01}$ or~$\lab{10}$ is in~$\LSemb{G_{uv}}{\calE_{uv}}$.
        Then there is an inner augmentation~$H_{uv}$ of~$G_{uv}$ 
        with label~$\lab{01}$ or~$\lab{10}$.
        In particular there is a degree-$1$ vertex~$w$ in the outer face of~$H_{uv}$.
        Now add edges~$uw$ and~$vw$, and replace each of $u,v$ that has degree~$2$ by the gadget from \cref{fig:gadget_unsaturated}.
        This yields a $2$-connected $3$-augmentation of $G_{uv}^+$.
        
        \item Lastly, to check whether~$\lab{11} \in \LSemb{G_{uv}}{\calE_{uv}}$, we build a graph~$G_{uv}^+$ by taking the $uv$-graph from \cref{fig:gadget_11}, embedding it into the outer face of~$\calE_{uv}$ and identifying the respective~$u$ and~$v$ vertices.
        If afterwards~$u,v$ have degree~$2$, we replace them by degree-$3$ vertices using the gadget from~\cref{fig:gadget_unsaturated}.
        Let~$\ell,r$ be the two degree-$2$ vertices as shown in \cref{fig:gadget_11}.
        Again, we claim that $\lab{11} \in \LSemb{G_{uv}}{\calE_{uv}}$ if and only if~$G_{uv}^+$ has a $2$-connected $3$-augmentation extending its embedding.
        
        If such a $3$-augmentation exists, removing all vertices from the $uv$-graph in \cref{fig:gadget_11} except for~$u,v,\ell,r$ (and possibly undoing the replacements of~$u,v$) yields an inner augmentation~$H_{uv}$ with label~$\lab{11}$ (because there are two degree-$1$ vertices and a $uv$-edge in the outer face of the embedding of~$H_{uv}$ would separate them into different faces).
        
        To construct a $2$-connected $3$-augmentation from an inner augmentation~$H_{uv}$ of~$G_{uv}$ with label~$\lab{11}$, consider a hypothetical $uv$-edge separating the two degree-$1$ vertices of~$H_{uv}$ into different faces.
        Add the graph from \cref{fig:gadget_11} as above in the embedding of~$H_{uv}$ where the $uv$-edge would have been.
        Then identify the two degree-$1$ vertices with~$\ell$ and~$r$ in a non-crossing way.
        This yields a $3$-augmentation of $G_{uv}$ after possibly replacing $u$ and/or $v$ by the gadget from \cref{fig:gadget_unsaturated}, as in the previous cases.
    \end{itemize}
    In all three cases, the existence of a $2$-connected $3$-augmentation of~$G_{uv}^+$ can be checked using \cref{prop:2con-3aug-fixed-multi-graph}.
    As only a constant number of vertices and edges were added to~$G_{uv}$ this check takes time in~$\mathcal{O}(n^2)$.
\end{proof}

We remark that in the proof of \cref{lem:label-check} for each label~$\lab{ab}$ we actually added the gadget for the embedded label set $\{\lab{ab}\}$ between vertices~$u,v$.
Thus, these three gadgets serve a twofold role in our algorithm.

\subparagraph{Algorithm for Variable Embedding.}
In order to decide whether a given $2$-connected planar graph~$G$ admits some planar embedding which admits a $2$-connected $3$-augmentation, we use the SPQR-tree~$T$ of~$G$.
Rooting~$T$ at some Q-vertex~$\rho$, the pertinent graph~$\pert(\mu)$ of a vertex~$\mu$ in~$T$ is a subgraph of~$G$.
Moreover, if~$u_\mu v_\mu$ is the virtual edge in $\skel(\mu)$ associated to the parent edge of~$\mu$, then $\pert(\mu)$ is a $uv$-graph (with~$u_{\mu}, v_{\mu}$ taking the roles of~$u,v$ in the $uv$-graph).
Now the variable label set~$\LSvar{\pert(\mu)}$ is a constant-size representation of all possible labels that any possible embedding of an inner augmentation of~$\pert(\mu)$ can have (having~$u_{\mu}$ and~$v_{\mu}$ on its outer face).
The remainder of this section describes how the variable label sets of all vertices in the SPQR-tree can be computed by a bottom-up dynamic program.
Here we need to distinguish whether we consider an~S- a~P- or an R-vertex of the SPQR-tree.

\begin{lemma}
    \label{lem:S-vertex}
    Let~$\mu$ be an S-vertex of the SPQR-tree with children~$\mu_1, \ldots, \mu_k$, such that each variable label set~$\LSvar{\pert(\mu_i)}$, $i=1,\ldots,k$, is non-empty and known.
    Then the variable label set~$\LSvar{\pert(\mu)}$ can be computed in time $\mathcal{O}(\| \skel(\mu) \|^2)$.
\end{lemma}

\begin{proof}
    Let~$uv$ be the virtual edge associated to the parent edge of~$\mu$.
    Further, let~$u_iv_i$ be the virtual edge associated to the tree edge~$\mu\mu_i$, $i=1,\ldots,k$.
    
    First, we remove the edge~$uv$ in~$\skel(\mu)$ to obtain once again a $uv$-graph~$G_{uv}$.
    As~$G_{uv}$ is a path, its planar embedding is unique.
    Replace each virtual edge~$u_iv_i$ by the gadget~$Y$ with embedding~$\calE_Y$ from \cref{lem:gadgets} that realizes the embedded label set $\LSemb{Y}{\calE_Y} = \LSvar{\pert(\mu_i)}$.
    Call the obtained graph~$G_{\mu}$.
    Second, we consider the vertices in~$G_{\mu}$ that belong to~$\skel(\mu)$ (i.e., those that have not been introduced by some gadget).
    Each~$w \in \skel(\mu)$ has degree~$2$ in $\skel(\mu)$.
    If~$\deg_G(w) = 3$ or if~$w$ is one of~$u,v$, then we replace~$w$ by the gadget shown in \cref{fig:gadget_unsaturated} to replace each such degree-$2$ vertex by four degree-$3$ vertices.
    This is necessary because if~$\deg_G(w) = 3$ in~$G$, we may not add in an augmentation additional edges to $w$ at any time.
    Additionally, if $w \in \{u,v\}$, then we consider possible new edges at vertex $w$ further upwards in the SPQR-tree and not here.
    
    By a slight abuse of notation, we still call the obtained graph~$G_{\mu}$ and its planar embedding as constructed~$\calE_{\mu}$.
    As before, we have that $\LSvar{\pert(\mu)} = \LSemb{G_{\mu}}{\calE_{\mu}}$.
    As every gadget has constant size, we have $\| G_{\mu} \| \in \mathcal{O}(\| \skel(\mu) \|)$.
    By \cref{lem:label-check}, we can compute $\LSemb{G_{\mu}}{\calE_{\mu}}$ in time $\mathcal{O}(\| \skel(\mu) \|^2)$.
\end{proof}

\begin{lemma}
    \label{lem:P-vertex}
    Let~$\mu$ be a P-vertex of the SPQR-tree with children~$\mu_1, \ldots, \mu_k$, such that each variable label set~$\LSvar{\pert(\mu_i)}$, $i=1,\ldots,k$, is non-empty and known.
    Then the variable label set~$\LSvar{\pert(\mu)}$ can be computed in time $\mathcal{O}(1)$.
\end{lemma}

\begin{proof}
    By definition,~$\skel(\mu)$ consists of two vertices and at least three parallel edges.
    Since~$\Delta(G) \leq 3$, we have that~$\skel(\mu)$ contains exactly three parallel edges.
    To be able to distinguish between these three edges, we write~$uv$ for the virtual edge associated to the parent edge of~$\mu$, and~$u_1v_1$,~$u_2v_2$ for the virtual edges associated to the tree edges~$\mu\mu_1$,~$\mu\mu_2$.
    
    As in (the proofs of) \cref{lem:R-vertex,lem:S-vertex}, we fix a planar embedding of $\skel(\mu)$ with edge~$uv$ on the outer face and then remove edge~$uv$ to obtain a $uv$-graph~$G_{uv}$.
    The planar embedding of~$G_{uv}$ is again unique.
    (We can again ignore the flipped embedding, as variable label sets are symmetric.)
    
    For both children~$\mu_i$ the variable label sets~$\LSvar{\pert(\mu_i)}$ are non-empty and known.
    We replace~$u_iv_i$, $i=1,2$, in~$G_{uv}$ by the gadget~$Y$ with embedding~$\calE_Y$ from \cref{lem:gadgets} that realizes the embedded label set $\LSemb{Y}{\calE_Y} = \LSvar{\pert(\mu_i)}$.
    We call the obtained graph~$G_\mu$ and~$\calE_{\mu}$ its planar embedding as constructed.
    
    As above, we have $\LSvar{\pert(\mu)} = \LSemb{G_{\mu}}{\calE_{\mu}}$ because the embedded label set of each gadget equals the variable label set of the pertinent graph replaced by it.
    Thus, we can again use \cref{lem:label-check} to compute $\LSemb{G_{\mu}}{\calE_{\mu}}$ and therefore also $\LSvar{\pert(\mu)}$.
    This takes constant time, because~$\skel(\mu)$ and each gadget has constant size.
\end{proof}

\begin{lemma}
    \label{lem:R-vertex}
    Let~$\mu$ be an R-vertex of the SPQR-tree with children~$\mu_1, \ldots, \mu_k$, such that each variable label set~$\LSvar{\pert(\mu_i)}$, $i=1,\ldots,k$, is non-empty and known.
    Then the variable label set~$\LSvar{\pert(\mu)}$ can be computed in time $\mathcal{O}(\| \skel(\mu) \|^2)$.
\end{lemma}

\begin{proof}
    As~$\mu$ is an R-vertex, its skeleton~$\skel(\mu)$ is $3$-connected.
    So by Whitney's Theorem~\cite{Whitney1933_UniqueEmbedding} $\skel(\mu)$ has a unique planar embedding up to flipping and the choice of the outer face.
    Let~$uv$ be the virtual edge in~$\skel(\mu)$ associated to the parent edge of~$\mu$.
    We choose an arbitrary planar embedding of~$\skel(\mu)$ with~$uv$ on the outer face and then remove the edge between~$u$ and~$v$ to obtain a $uv$-graph~$G_{uv}$.
    Note that the induced planar embedding of~$G_{uv}$ is now unique (apart from its flipped embedding, which we can ignore because variable label sets are symmetric).
    
    For each virtual edge~$u_iv_i$, $i=1,\ldots,k$, associated to a tree edge~$\mu\mu_i$ we know the label set~$\LSvar{\pert(\mu_i)}$.
    Replace the virtual edge $u_iv_i$ in~$G_{uv}$ by the gadget~$Y$ with embedding~$\calE_Y$ from \cref{lem:gadgets} with~$\LSemb{Y}{\calE_Y} = \LSvar{\pert(\mu_i)}$.
    Let~$G_{\mu}$ be the obtained graph and~$\calE_{\mu}$ its planar embedding as constructed above.
    
    Because the embedded label set of a gadget~$Y$ with fixed embedding~$\calE_{Y}$ equals the variable label set of the corresponding subgraph~$\pert(\mu_i)$, it follows that $\LSvar{\pert(\mu)} = \LSemb{G_{\mu}}{\calE_{\mu}}$.
    This equivalent reformulation of the variable label set of $\pert(\mu)$ as the embedded label set of $G_\mu$ is the key insight.
    Each gadget has constant size, so $\| G_{\mu} \| \in \mathcal{O}(\| \skel(\mu) \|)$.
    Thus we can compute~$\LSemb{G_{\mu}}{\calE_{\mu}}$ in time $\mathcal{O}(\| \skel(\mu) \|^2)$ using \cref{lem:label-check}.
\end{proof}

\Cref{lem:S-vertex,lem:P-vertex,lem:R-vertex} compute the variable label set of an inner vertex of the SPQR-tree, requiring that the variable label sets of its children are non-empty.
If this condition is not satisfied, i.e., at least one vertex~$\mu$ has $\LSvar{\pert(\mu)} = \emptyset$, then the following \lcnamecref{lem:empty-label-set} applies:

\begin{lemma}
    \label{lem:empty-label-set}
    If~$\LSvar{\pert(\mu)} = \emptyset$ for some vertex~$\mu$ of the SPQR-tree~$T$ of~$G$, then~$G$ has no $2$-connected $3$-augmentation.
\end{lemma}

\begin{proof}
    Assuming that~$G$ has a $2$-connected $3$-augmentation~$H$, we shall show that we have $\LSvar{\pert(\mu)} \neq \emptyset$ for every vertex~$\mu$ of~$T$.
    If~$\mu$ is the root, let~$u,v$ be the two unique vertices in~$\skel(\mu)$ (because $\mu = \rho$ is a Q-vertex).
    If~$\mu$ is not the root, let~$u,v$ be the endpoints of the virtual edge associated to the parent edge of~$\mu$.

    By the definition of labels, $\LSvar{\pert(\mu)} \neq \emptyset$ if there is some inner augmentation of~$\pert(\mu)$ for at least one of its planar embeddings with~$u,v$ on its outer face.
    But the $2$-connected $3$-augmentation~$H$ of~$G$ induces an inner augmentation of~$\pert(\mu)$ as follows:
    Let~$\calE_H$ be a planar embedding of~$H$ with outer edge~$e_\rho$ and~$\calE_G$ its restriction to~$G$.
    Recall that then~$u,v$ are outer vertices of~$\pert(\mu)$ in~$\calE_G$.
    Consider the embedded subgraph of~$H$ consisting of~$\pert(\mu)$ and all vertices and edges of~$H$ inside inner faces of~$\pert(\mu)$ in~$\calE_G$.
    For each vertex~$w \neq u,v$ on the outer face of~$\pert(\mu)$ in~$\calE_G$ incident to an edge of~$H$ in the outer face of~$\pert(\mu)$, we add a new pendant edge at~$w$ into the outer face of~$\pert(\mu)$ in~$\calE_G$.
    The resulting graph is an inner augmentation of~$\pert(\mu)$ and hence $\LSvar{\pert(\mu)} \neq \emptyset$.
\end{proof}

\medskip

Now that we considered S-, P- and R-vertices, we are finally set up to prove \cref{prop:variable-embedding}.
There we claim that we can decide in polynomial time whether a $2$-connected planar graph~$G$ with~$\Delta(G) \leq 3$ has a $2$-connected $3$-augmentation.

\begin{proof}[Proof of \cref{prop:variable-embedding}]
    As mentioned above, we use bottom-up dynamic programming on the SPQR-tree~$T$ of~$G$ rooted at an arbitrary Q-vertex~$\rho$ corresponding to an edge~$e_\rho$ in~$G$.

    The base cases are the leaves of~$T$, all of which are Q-vertices.
    The variable label set of a leaf~$\mu$ is~$\LSvar{\pert(\mu)} = \{\lab{00}\}$:
    $\pert(\mu)$ is just a single edge and the only inner augmentation of~$\pert(\mu)$ is $\pert(\mu)$ itself, and as such has label $\lab{00}$.

    Now let~$\mu$ be an inner vertex of~$T$ and thus be either an S-, a P- or an R-vertex.
    All its children~$\mu_1, \ldots, \mu_k$ have already been processed and their variable label sets~$\LSvar{\pert(\mu_i)}$ are known.
    Then the variable label set~$\LSvar{\pert(\mu)}$ can be computed in time $\mathcal{O}(\| \skel(\mu) \|^2)$ (which is actually $\mathcal{O}(1)$ in case of a P-vertex) by \cref{lem:S-vertex,lem:P-vertex,lem:R-vertex}.
    To apply these lemmas, we need to guarantee that the variable label sets~$\LSvar{\pert(\mu_i)}$ of the children are non-empty.
    If this is not the case, then by \cref{lem:empty-label-set}~graph $G$ has no $2$-connected $3$-augmentation and we can stop immediately.

    It remains to consider the root~$\rho$ of the SPQR-tree.
    Recall that~$\pert(\rho) = G$.
    Following the setup of \cref{lem:3aug-if-00-label}, let~$x,y$ be the two unique vertices of~$\skel(\rho)$ and~$xy$ be the unique non-virtual edge, i.e., the edge $e_\rho = xy$ of~$G$.
    Let~$G_{uv}$ be the $uv$-graph obtained from $G = \pert(\rho)$ by deleting $e_\rho = xy$ and adding two new pendant edges $ux,vy$.
    Note that~$x$ and~$y$ have the same degree in~$G_{uv}$ as in~$G$.
    By \cref{lem:3aug-if-00-label},~$G$ has a $2$-connected $3$-augmentation if and only if $\lab{00} \in \LSvar{G_{uv}}$.
    
    To check whether $\lab{00} \in \LSvar{G_{uv}}$, let~$\mu$ be the unique child of~$\rho$.
    Thus we have $\pert(\mu) = G - e_\rho$.
    We have already computed $\LSvar{\pert(\mu)}$ and can assume by \cref{lem:empty-label-set} that it is non-empty.
    Consider the gadget~$Y$ with embedding~$\calE_Y$ from \cref{lem:gadgets} such that $\LSemb{Y}{\calE_Y} = \LSvar{\pert(\mu)}$.
    Let~$u'$ and~$v'$ denote the two degree-$1$ vertices in~$Y$.
    If both~$x$ and~$y$ have degree~$3$ in~$G$ (hence also in~$G_{uv}$), then $\LSvar{G_{uv}} = \LSvar{\pert(\mu)} = \LSemb{Y}{\calE_Y}$ and we already know whether or not~$\lab{00}$ is contained in these label sets.
    
    If~$x$ has degree~$2$ in~$G$ (hence also degree~$2$ in~$G_{uv}$, while degree~$1$ in~$\pert(\mu)$), then~$x$ receives a new edge in inner augmentations of~$G_{uv}$ but not in inner augmentations of~$\pert(\mu)$.
    For~$Y$ to model~$\LSvar{G_{uv}}$ instead of~$\LSvar{\pert(\mu)}$, we subdivide in~$Y$ the edge at~$u'$ by a new vertex~$x'$.
    Similarly, if~$y$ has degree~$2$ in~$G$, we subdivide in~$Y$ the edge at~$v'$.
    For the resulting graph~$Y'$ with embedding~$\calE_{Y'}$ it follows that $\LSvar{G_{uv}} = \LSemb{Y'}{\calE_{Y'}}$ and we can check whether~$\lab{00}$ is contained in these label sets by calling \cref{lem:label-check} on~$Y'$ with embedding~$\calE_{Y'}$.
    This takes constant time, as~$Y'$ has constant size.
    
    The overall runtime is the time needed to construct the SPQR-tree plus the time spent processing each of its vertices.
    Gutwenger and Mutzel~\cite{Gutwenger2001_SPQRLinear} show how to construct the SPQR-tree in time~$\mathcal{O}(n)$.
    The time for the dynamic program traversing the SPQR-tree~$T$ is
    \[
        \mathcal{O}\bigl(
            \sum\limits_{\mu \in V(T)} \| \skel(\mu) \|^2
        \bigr)
        \subseteq
        \mathcal{O}\bigl(\bigl(
            \sum\limits_{\mu \in V(T)} \| \skel(\mu) \|
        \bigr)^2\bigr)
        \subseteq
        \mathcal{O}(n^2)
        \text{,}
    \]
    where the first step uses that for a set of positive integers the sum of their squares is at most the square of their sum, and the second step uses that the SPQR-tree has linear size.
\end{proof}

\section{Connected~3-Augmentations for a Fixed Embedding}
\label{sec:1con_3aug_fixed_embedding}

In this section, we find connected $3$-augmentations for arbitrary (and possibly disconnected) input graphs.
Refer to the second column of the table in \cref{fig:overview}.
For the variable embedding setting, the problem can be solved in linear time, see \cref{prop:wlog-connected}.
We now present a quadratic-time algorithm for the fixed embedding setting.
By~\cref{obs:connected_3-augmentation}, this is equivalent to deciding whether a given embedding can be extended to a subcubic planar connected graph.

\begin{lemma}
\label{lem:connected_augmentation}
    Let~$G$ be a planar subcubic graph with an embedding~$\calE$ and $\Delta(G) \leq 3$.
    Let $n$ be the number of vertices in $G$.
    Then we can compute, in time~$\calO(n^2)$, a connected subcubic planar supergraph~$H$ of~$G$ extending~$\calE$, or conclude that none exists. 
\end{lemma}

We remark that the proof of \cref{lem:connected_augmentation} is essentially a simpler version of the proof of \cref{lem:2con_fixed_mindeg2}.
Yet, as we now only need to find a \emph{subcubic} planar connected supergraph, we do not need to assign exactly $3-\deg_{G}(x)$ new edges to a vertex~$x$ of degree less than~$3$.
In fact, such a vertex can be incident to any number of new edges, ranging from~$0$ to~$3-\deg_{G}(x)$. 
The same replacement rules as in \cref{lem:preprocessing} could be applied in order to reduce the problem to a input graph~$G$ with~$\delta(G) \geq 2$.
However, this will not be necessary as we can handle in this section vertices of degree~$0$ and~$1$ in the same way as vertices of degree~$2$.

\begin{proof}
    Using the same ideas as in the proof of \cref{lem:2con_fixed_mindeg2}, we reduce the problem of finding a connected supergraph which extends~$\calE$, to an instance of the \GeneralizedFactor problem~$A$ which fulfills the necessary condition to apply an~$\calO(n^2)$-time algorithm by~\Sebo, see \cref{thm:generalized-factor-sebo}. 

    We construct the supergraph~$H$ of~$G$ by adding vertices and edges to some faces of~$\calE$.
    Thus, the obtained embedding of~$H$ extends~$\calE$. 
    
    In order to obtain a connected supergraph, faces of~$\calE$ which are incident to at least two connected components \emph{must} contain new edges (and possibly vertices).
    We call the set of these the \emph{component-connecting faces}~$F_{\compconn}$. 

    For a component-connecting face~$f \in F_{\compconn}$, we denote by~$G_f$ the subgraph of~$G$ on the vertices and edges incident to~$f$ (using the same notation as in the proof of~\cref{lem:2con_fixed_mindeg2}). 

    The~\GeneralizedFactor instance~$A$ is a bipartite graph with bipartition classes~$\calV$ and $\calF$. 
    Here, $\calV \coloneqq \{v \in V \mid \deg_G(v) \leq 2\}$ contains all vertices of~$G$ that have degree less than~$3$ (and hence require new edges). 
    Vertices in~$\calF$ represent component-connecting faces of~$\calE$. 
    For each component-connecting face~$f \in F_{\compconn}$, we add all components of~$G_f$ as vertices to~$\calF$. 
    (If there are two faces~$f,g$ in $\calE$ such that~$G_f$ and~$G_g$ contain two components corresponding to the same subgraph of~$G$, then~$\calF$ contains two such vertices: one corresponding to the component of~$G_f$, and another to the component of~$G_g$.)

    Each~$c \in \calF$ corresponds to a component of~$G_f$ for some face of~$\calE$. 
    In the graph~$A$, $c$ is incident to all $v \in \calV$ which are incident to the component~$c$ in~$\calE$.
    It remains to assign a set~$B(x) \subseteq \{0,1, \dots, \deg_{A}(x)\}$ of possible degrees to each vertex~$x \in \calV \cup \calF$:
    \[
        B(x) \coloneqq 
        \begin{cases}
            \{k \in \N_0 \mid k \leq \min(3-\deg_{G}(x), \deg_A(x))\}, & \text{if $x \in \calV$} \\
            \{1, 2, 3, \dots, \deg_{A}(x)\}, & \text{if $x \in \calF$,}
        \end{cases}
    \]
    i.e., each vertex $x \in \calV$ can be incident to up to $3-\deg_{G}(x)$ edges, and each $c \in \calF$ is incident to at least one edge in any $B$-factor.
    
    Observe that the order and size of~$A$ are linear in~$n$:
    Every vertex~${u \in \calV}$ is incident to at most two faces of~$\calE$ as $\deg_G(u) \leq 2$. 
    Since~$u$ belongs to at most one component of~$G_f$ for each component-connecting face~$f$, we obtain $\abs{\calF} \leq 2n$ and $\deg_{A}(v) \leq 2$ for every~$v \in \calV$. 
    In particular, it follows that $\abs{E(A)} \leq 2n$.
    Note that~$A$ can be computed in linear time.

    Using a similar argument as in~\cref{lem:2con_fixed_mindeg2}, we observe that~$A$ admits a~$B$-factor if and only if~$G$ has a connected subcubic planar supergraph~$H$ extending~$\calE$. 

    Recall that~$\deg_{A}(v) \leq 2$ for all $v \in \calV$.
    For $v \in \calV$ with $\deg_G(v) = 2$ and~$\deg_A(v) = 2$, we have $B(v) = \{0,1\}$. 
    The set~$B(v)$ then excludes exactly one possible degree.
    In fact, this is the only case when~$B(x)$ does not include all possible degrees for a vertex~$x \in \calV \cup \calF$.
    We can therefore apply the algorithm by~\Sebo (see \cref{thm:generalized-factor-sebo}) and compute a~$B$-factor of~$A$ in time~$\calO(n^2)$.
\end{proof}

\begin{remark}
    The \GeneralizedFactor instance we constructed in the proof above can be reduced to an instance of \MaximumFlow.
    In theory\footnote{The algorithm in~\cite{Chen2023_AlmostLinearTime} is randomized and not practical. Our combinatorial, but asymptotically slower, algorithm is preferable here.}, this achieves a better, namely almost-linear~\cite{Chen2023_AlmostLinearTime}, runtime.
    The graph of the flow instance is obtained by adding two vertices~$s$ and~$t$ (namely the source and sink) and connecting~$s$ to all vertices in~$\calV$ and~$t$ to all vertices in~$\calF$. 
    Edges incident to~$s$ are outgoing, edges incident to~$t$ are incoming.
    Edges between~$\calV$ and~$\calF$ are oriented from $\calV$ to~$\calF$.
    The edge capacities of edges~$sx$ with $x \in \calV$ encode the sets $B(x)=\{0,1,2,\dots, 3-\deg_G(x)\}$.
    This is achieved by setting the edge capacity to~$3-\deg_G(x)$. 
    All other edges have a capacity of~$1$.
    It can be easily verified that the obtained graph admits an $s$--$t$-flow of value at least~$\abs{\calF}$ if and only if~$A$ admits a~$B$-factor.
\end{remark}

\section{\NP-Hardness for 3-Connected 3-Augmentations}
\label{sec:3con_3aug_variable_embedding}

In this section, we prove that deciding whether a given planar graph~$G$ admits a $3$-connected $3$-augmentation is \NP-complete. In particular, we show that the problem remains \NP-complete when restricted to connected graphs~$G$. This implies the \textsf{NPC}-results represented in the fourth column of the table in~\cref{fig:overview}, corresponding to Statement~\ref{itm:main_theorem_3con_variable} of \cref{thm:main_theorem}.

We reduce from the \NP-complete problem \textsc{Planar-Monotone-3SAT}~\cite{deBerg2010_Planar3SAT}.
In intermediate steps of the proof, we obtain graphs with vertices of degree greater than~$3$. 
As we are interested in $3$-connected $3$-augmentations, we need a transformation that replaces such a vertex by a small gadget in order to obtain a maximum degree of at most~$3$ while preserving $3$-connectivity. 
This is realized by \emph{wheel-extensions}.
Consider a cycle~$C_{\ell}$ of length~$\ell \geq 3$ and a path~$P_2$.
We denote by $C_{\ell} \square P_2$ their product. 
For an integer~$\ell \geq 3$, let~$W_\ell$ be the graph obtained from $C_\ell \square P_2$ by subdividing each edge in one cycle~$C_\ell$ exactly once.
See \cref{fig:wheel_extension_wheel} for an illustration.
Consider a planar graph~$G$ with an embedding~$\calE$, and a vertex~$v \in V(G)$ with~$\deg_G(v) = \ell \geq 3$.
A \emph{wheel-extension at~$v$} is the graph and embedding obtained by replacing~$v$ with~$W_\ell$, and by attaching $v$'s incident edges to the subdivision vertices of~$W_\ell$ in a one-to-one non-crossing way.
See \cref{fig:wheel_extension_extension}.

\begin{figure}[htb]
    \centering
    \begin{subfigure}{0.3\textwidth}
        \centering
        \includegraphics[page=2]{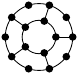}
        \caption{}
        \label{fig:wheel_extension_wheel}
    \end{subfigure}
    \hspace{4mm}
    \begin{subfigure}{0.6\textwidth}
        \centering
        \includegraphics[page=3]{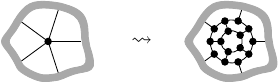}
        \caption{}
        \label{fig:wheel_extension_extension}
    \end{subfigure}
    \caption{
        \subref{fig:wheel_extension_wheel} The graph~$W_5$ obtained from $C_5 \square P_2$ through subdivision.
        \subref{fig:wheel_extension_extension} A Wheel-extension.
    }
    \label{fig:wheel_extension}
\end{figure}

\begin{observation}
    \label{obs:wheel_extension}
    Let~$G$ be a graph (possibly with multi-edges, but no loops), let $v \in V(G)$ be a vertex with~$\deg_G(v) \geq 3$, and let~$G'$ be obtained from~$G$ by a wheel-extension at~$v$.
    Then~$\theta(G') \geq \min\{\theta(G), 3\}$.
\end{observation}

\begin{proof}
    If~$\theta(G') \leq 2$ (otherwise there is nothing to show), let~$S$ be an edge-cut of size~$\theta(G') \leq 2$ in~$G'$.
    As~$S$ is minimal,~$S$ does not consist of the two edges at a subdivision vertex of~$W_\ell$.
    Thus, as~$C_\ell \square P_2$ is~$3$-connected, it follows that~$S \cap E(W_\ell) = \emptyset$.
    But then,~$S$ is also an edge-cut in~$G$ and hence~$\theta(G) \leq |S| = \theta(G')$, as desired.
\end{proof}

Recall that an embedding of any $3$-connected $3$-augmentation~$H$ induces an embedding~$\calE$ of~$G$. For convenience, we call such a pair~$(H, \calE)$ a \emph{solution} for~$G$.
Let us also define a \emph{$(\leq2)$-subdivision} of a graph~$R$ to be the result of subdividing each edge in~$R$ with up to two vertices.
Note that, if~$R$ is $2$-connected, then so is every $(\leq2)$-subdivision of~$R$.

\begin{lemma}
    \label{lem:3con_fixed_augmentation}
    Let~$G$ be a graph obtained from a $(\leq2)$-subdivision~$R_2$ of a $3$-connected planar graph~$R$ by attaching a degree-$1$ vertex to each subdivision vertex.
    Then~$G$ admits a solution~$(H, \calE)$ if and only if no face of~$\calE$ has exactly one or two incident degree-$1$ vertices.
\end{lemma}

\begin{proof}
    First, assume that~$(H, \calE)$ is a solution for~$G$.
    Assume for the sake of contradiction that~$f$ is a face of~$\calE$ incident to a set~$S$ of exactly one or two degree-$1$ vertices of~$G$.
    As~$H$ is $3$-regular, each vertex in~$S$ is incident to two new edges in~$f$.
    But then,~$S$ forms a vertex-cut of cardinality at most~$2$ in~$H$; a contradiction to~$H$ being $3$-connected.

    For the other direction, let~$\calE$ be an embedding of~$G$ in which no face has exactly one or two incident degree-$1$ vertices.
    Our task is to find a solution~$(H, \calE)$ for~$G$, i.e., to insert new vertices and new edges into the faces of~$\calE$ to obtain a $3$-connected $3$-regular planar graph~$H$.
    
    To this end, consider any face~$f$ of~$\calE$.
    If~$f$ has no incident degree-$1$ vertices of~$G$, we insert nothing in~$f$.
    Otherwise,~$f$ has at least~$\ell \geq 3$ incident degree-$1$ vertices, and we identify all these vertices into one vertex~$v_f$ of degree~$\ell$.
    Let~$H_1$ be the planar graph we obtain by doing this for all faces of~$\calE$.
    Clearly,~$H_1$ is planar, $\delta(H_1) \geq 3$, and~$R_2 \subset H_1$.
    
    We claim that~$H_1$ is $3$-edge-connected, i.e.,~$\theta(H_1) \geq 3$.
    First,~$H_1$ is connected, as~$R_2$ is connected.
    It remains to show that the plane dual~$H_1^*$ of~$H_1$ has no loops (i.e.,~$H_1$ has no bridges) and no pairs of parallel edges (i.e.,~$H_1$ has no $2$-edge-cuts).
    For this, consider any edge~$e^*$ of~$H_1^*$ and its primal edge~$e$ of~$H_1$.
    If~$e \notin E(R_2)$, then~$e$ is incident to a vertex~$v_f$ of degree~$\ell \geq 3$ in a face~$f$ of~$\calE$.
    In this case~$e^*$ is neither a loop nor has a parallel edge in~$H_1^*$.

    If~$e \in E(R_2)$, then~$e^*$ is not a loop, since~$R_2$ is $2$-connected.
    It remains to rule out that two edges~$e_1, e_2 \in E(R_2)$ form a $2$-edge-cut, i.e., their dual edges~$e_1^*,e_2^*$ in~$H_1^*$ are parallel.
    Let~$f, f'$ be the two faces of~$\calE$ incident to~$e_1$ and~$e_2$.
    As~$R$ is $3$-connected,~$e_1$ and~$e_2$ both originate from the subdivision(s) of the same edge~$e_R$ of~$R$.
    Consider a subdivision vertex of~$R_2$ between~$e_1$ and~$e_2$.
    Let~$v$ be its new neighbor in~$H_1$; say~$v = v_f$ for face~$f$.
    Then~$v_f$ has at least two further neighbors, at least one of which is not a subdivision vertex of~$e_R$, because~$e_R$ is subdivided at most twice.
    But then in the dual~$H_1^*$, the edges~$e^*_1$ and~$e^*_2$ are incident to different vertices inside~$f$; hence are not parallel.

    Finally, we apply a wheel-extension to every vertex~$v_f$, resulting in a planar $3$-regular graph~$H$.
    Further,~$H$ contains~$G$ as a subgraph and \cref{obs:wheel_extension} yields~$\theta(H) \geq \min(\theta(H_1),3) = 3$.
    In other words,~$H$ is the desired $3$-connected $3$-augmentation of~$G$.
\end{proof}

By \cref{lem:3con_fixed_augmentation}, any graph~$G$ as described in the lemma admits a $3$-connected $3$-augmentation if and only if it admits an embedding~$\calE$ with no face incident to exactly one or two degree-$1$ vertices.
Testing such graphs for such embeddings, however, turns out to be \NP-complete.

\begin{theorem}
    \label{thm:NP_complete_3conn_variable}
    Deciding whether a given graph is a subgraph of a $3$-regular $3$-connected planar graph is \NP-complete.
\end{theorem}

\begin{proof}
    First, we show that the problem is in~\NP.
    Let~$G$ be a graph that admits a $3$-connected $3$-augmentation~$H$.
    We need to show that~$G$ also admits a $3$-connected $3$-augmentation whose size is polynomial in the size of~$G$.
    To this end, consider the subgraph~$N$ of~$H$ induced by all new vertices.
    In~$H$, contract each connected component of~$N$ into a single vertex, keeping parallel edges but removing loops.
    The resulting graph~$H'$ is planar and $3$-edge-connected, since so was~$H$.
    Note in particular that the vertices we obtained by contraction have degree at least~$3$. 
    Next, we apply a wheel-extension to each vertex obtained from the contractions that has degree larger than~$3$.
    By \cref{obs:wheel_extension}, the resulting graph~$H''$ is~$3$-regular and $3$-edge-connected, hence also~$3$-connected; in particular, a~$3$-connected $3$-augmentation of~$G$.
    Moreover, each vertex in~$H''$ has distance at most three to some vertex of~$G$ and the maximum degree is bounded by~$3$, and thus~$H''$ has only~$O(|V(G)|)$ many vertices.
    Thus, our decision problem is in~\NP. 
    
    \medskip
    \noindent
    To show \NP-hardness, we reduce from \textsc{Planar-Monotone-3SAT}.
    An instance of \textsc{Planar-Monotone-3SAT} is a monotone \textsc{3SAT}-formula~$\Psi$ together with its bipartite variable-clause incidence graph~$I_\Psi$ and a planar embedding~$\calE_\Psi$ of $I_\Psi$.
    Each clause in~$\Psi$ contains either only positive literals (then the clause is called \emph{positive}) or only negative literals (the clause is called \emph{negative}).
    Moreover,~$\Psi$ and the given embedding~$\calE_\Psi$ satisfy the following:
    \begin{itemize}
        \item Each variable lies on the $x$-axis and no edge crosses the $x$-axis.
        \item Positive clauses lie above the $x$-axis, negative clauses lie below the $x$-axis.
        \item Each clause has at most three literals.
    \end{itemize}
    It is known that \textsc{Planar-Monotone-3SAT} is \NP-complete~\cite{deBerg2010_Planar3SAT}.
    Note that the problem remains \NP-hard if we assume that each clause contains exactly three literals and that each variable appears in at least one positive and at least one negative clause.
    The latter can be achieved by the following reduction rule:
    Any variable that appears only in positive (negative) form can be set to \textsc{true} (\textsc{false}).

    \begin{figure}[tb]
        \centering
        \begin{subfigure}[b]{0.25\textwidth}
            \centering
            \includegraphics[page=6]{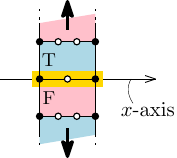}
            \caption{variable gadget}
            \label{fig:3SAT_gadgets_variable}
        \end{subfigure}
        \hfill
        \begin{subfigure}[b]{0.4\textwidth}
            \centering
            \includegraphics[page=7]{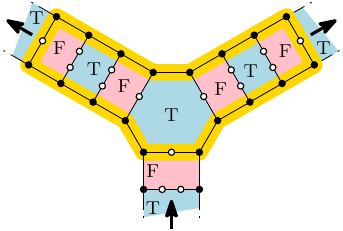}
            \caption{positive splitter gadget}
            \label{fig:3SAT_gadgets_splitter}
        \end{subfigure}
        \hfill
        \begin{subfigure}[b]{0.3\textwidth}
            \centering
            \includegraphics[page=8]{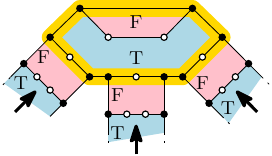}
            \caption{positive clause gadget}
            \label{fig:3SAT_gadgets_clause}
        \end{subfigure}
        \caption{Gadgets used in the \NP-hardness reduction.}
        \label{fig:3SAT_gadgets}
    \end{figure}

    Next, we construct a planar graph~$G_\Psi$ that admits a $3$-connected $3$-augmentation if and only if there exists a truth assignment of the variables in~$\Psi$ satisfying all clauses.
    The graph~$G_\Psi$ is obtained from the embedding~$\calE_\Psi$ of~$I_\Psi$ using the gadgets illustrated in \cref{fig:3SAT_gadgets}.
    In particular, we replace each variable by a copy of the variable gadget in \cref{fig:3SAT_gadgets_variable}.
    As illustrated by the arrows in \cref{fig:3SAT_gadgets}, these gadgets form the starting point of one \emph{upper} and one \emph{lower corridor} for each variable.
    Each corridor consists of alternating red (standing for false) and blue (standing for true) faces.
    Above the $x$-axis, for each variable with~$k$ occurrences in positive clauses, we use~$k-1$ positive splitter gadgets, see \cref{fig:3SAT_gadgets_splitter}, to split its upper corridor into~$k$ upper corridors.
    Then we replace each positive clause by a copy of the positive clause gadget in \cref{fig:3SAT_gadgets_clause}, which forms the end of one corridor for each variable appearing in the clause.
    These corridors are routed without overlap and entirely above the $x$-axis by following the given embedding~$\calE_\Psi$ of~$I_\Psi$.
    We proceed symmetrically below the $x$-axis, with red and blue swapped, using otherwise isomorphic negative splitter gadgets and negative clause gadgets.
    See \cref{fig:3SAT_example} for a full example.
    
    \begin{figure}[tb]
        \centering
        \includegraphics[page=1]{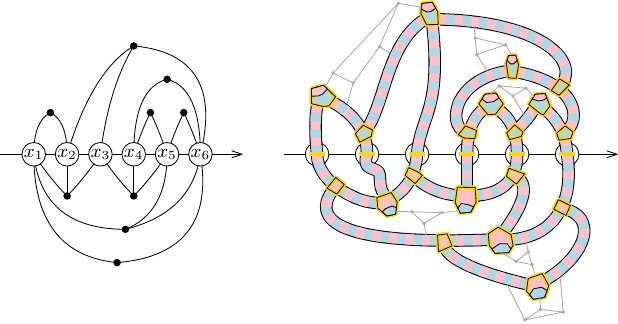}
        \caption{
            Illustration of a \textsc{Planar-Monotone-3SAT} embedding~$\calE_\Psi$ and a corresponding graph~$G_\Psi$.
            Extra vertices and edges added for the $3$-connectivity of~$R$ are shown in gray.
        }
        \label{fig:3SAT_example}
    \end{figure}
    
    The resulting graph is a $(\leq2)$-subdivision of a $3$-regular planar graph~$R$.
    We refer to the vertices of~$R$ as black vertices, and the subdivision vertices as white vertices.
    Each of the white (subdivision) vertices is incident to one red and one blue face, while each black vertex (of~$R$) is incident to an \emph{uncolored} (neither red nor blue) face.
    By adding additional vertices into uncolored faces and connecting these to incident edges, we modify~$R$ to obtain a $3$-connected $3$-regular planar graph~$R'$ which still has a $(\leq2)$-subdivision~$R_2$ including all gadgets and corridors.
    As~$R'$ is $3$-connected, the plane embedding of~$R_2$ is unique (up to the choice of the outer face).
    Finally, we attach a degree-$1$ vertex to each subdivision vertex, which completes the construction of~$G_\Psi$.
    See \cref{fig:3SAT_gadgets_detail}.

    \begin{figure}[tb]
        \centering
        \includegraphics[page=10]{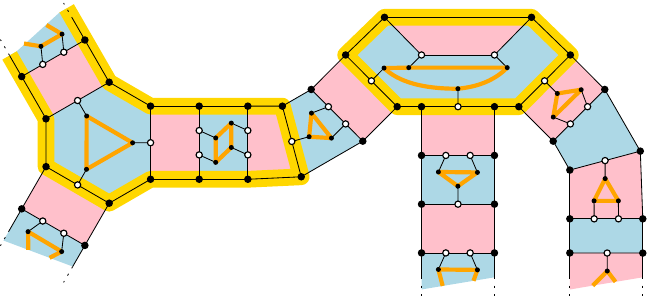}
        \caption{
            Part of the graph~$G_\Psi$ together with a $3$-connected $3$-augmentation in orange.
            This corresponds to a clause with two true variables (left with a splitter gadget, and middle) and one false variable (right).
        }
        \label{fig:3SAT_gadgets_detail}
    \end{figure}

    By \cref{lem:3con_fixed_augmentation},~$G_\Psi$ admits a $3$-connected $3$-augmentation if and only if~$G_\Psi$ admits an embedding~$\calE$ in which no face has exactly one or two incident degree-$1$ vertices.
    Since the embedding of the subgraph~$R_2$ of~$G_\Psi$ is fixed, such embedding~$\calE$ exists if and only if for each subdivision vertex we can choose either the incident red or the incident blue face in such a way that no face is chosen exactly once or twice.
    Except for the two highlighted faces in the clause gadgets, any pair of neighboring red and blue faces has in total at most five subdivision vertices.
    Thus, for each variable either all blue faces in all corridors are chosen, corresponding to the variable being set to true, or all red faces in all corridors are chosen, corresponding to the variable being set to false.
    In each positive clause gadget, the highlighted red face has only two incident subdivision vertices and hence cannot be chosen at all.
    Thus, the blue face of each positive clause gadget is chosen by two subdivision vertices, and thus must be chosen by at least one further subdivision vertex at the end of a variable corridor.
    This means that in each positive clause, at least one variable must be set to true.
    See again \cref{fig:3SAT_gadgets_detail} for an illustration.
    Symmetrically, at least one variable in each negative clause must be set to false, i.e., we have a satisfying truth assignment for~$\Psi$.
    In the same way, we obtain from a satisfying truth assignment for~$\Psi$ a valid choice for each subdivision vertex, i.e., an embedding of~$G_\Psi$ as required by \cref{lem:3con_fixed_augmentation}.

    To summarize, we obtain a planar graph~$G_\Psi$ that admits a $3$-connected $3$-augmentation if and only if the \textsc{Planar-Monotone-3SAT}-formula~$\Psi$ is satisfiable.
    The size of~$G_\Psi$ is polynomial in the size of~$\Psi$.
\end{proof}

\begin{remark}\label{rem:NPC-extension}
    The $(\leq2)$-subdivision in the above reduction behaves quite similar to~$G_\Psi$.
    The only problem is that a face~$f$ may have been chosen by exactly \emph{two} incident degree-$2$ vertices to contain their third (new) edge without creating a $2$-edge-cut; namely, with a direct edge.
    Thus, the above reduction also yields \NP-completeness of recognizing \emph{induced} subgraphs of $3$-connected $3$-regular planar graphs, even for~$2$-connected inputs with a unique embedding.
\end{remark}

\section{Discussion and Open Problems}
\label{sec:discussion}

Consulting the table in \cref{fig:overview}, our results show that for $k \leq 2$ finding $k$-connected $3$-augmentations is possible in polynomial time, both in the variable and the fixed embedding setting.
On the other hand, \cref{thm:NP_complete_3conn_variable} shows that finding $3$-connected $3$-augmentations is \NP-complete in the variable embedding setting, even if the input graph is connected.
The case of a fixed embedding and/or a $2$-connected input graph remains open.
We suspect these cases for $3$-connected $3$-augmentations to be \NP-complete as well.
For one thing, the graphs in our reduction in \cref{sec:3con_3aug_variable_embedding} are ``almost $2$-connected'' and have ``almost a unique embedding'', as discussed in \cref{rem:NPC-extension}.

\begin{figure}[ht]
    \centering
    \begin{subfigure}[b]{0.4\textwidth}
        \centering
        \includegraphics[page=2]{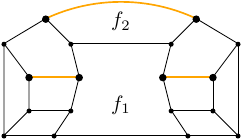}
        \caption{}
    \end{subfigure}
    \begin{subfigure}[b]{0.4\textwidth}
        \centering
        \includegraphics[page=3]{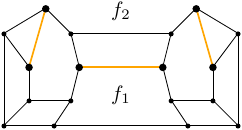}
        \caption{}
    \end{subfigure}
        \caption{Two $3$-connected $3$-augmentations of the same $2$-connected graph with fixed embedding.}
        \label{fig:2-connected}
\end{figure}

Additionally, finding $3$-connected $3$-augmentations for fixed embedding seems to crucially require a coordination among the new edges, which cannot be modeled as a \GeneralizedFactor problem with gaps of length~$1$.
For example, if the input graph $G$ is $2$-connected (but not already $3$-connected) with a fixed embedding, then there is an edge-cut of size~$2$.
See \cref{fig:2-connected} for an illustration.
To establish $3$-connectivity in a $3$-augmentation $H \supseteq G$, we must connect both sides of the cut through one of the two incident faces $f_1,f_2$, requiring both sides to coordinate and agree on which of~$f_1,f_2$ to choose.

\medskip

Our work was motivated by the complexity of the \textsc{3-Edge-Colorability}-problem.
We showed how to test in polynomial time whether a planar graph~$G$ is a subgraph of some $2$-connected cubic planar graph.
This yields a polynomial-time algorithm to recognize subgraphs of bridgeless cubic planar graphs.
Recall that such subgraphs admit proper $3$-edge-colorings.
(This follows from the Four-Color-Theorem~\cite{Appel1977_4Color1,Appel1977_4Color2} and the work of Tait~\cite{Tait1880_Bridgeless}.)
However, there are $3$-edge-colorable planar graphs with no bridgeless $3$-augmentation; $K_{2,3}$ is an easy example.
For another class of examples, consider for instance any $3$-connected $3$-regular plane graph~$G$ (that is, the dual of a plane triangulation) and subdivide (with a new degree-$2$ vertex each) any set of at least two edges, where no two of these are incident to the same face of~$G$ (so their dual edges form a matching in the triangulation).
The resulting graph~$G'$ has only one embedding (up to the choice of the outer face and flipping) and clearly no bridgeless $3$-augmentation.
On the other hand, \cref{conj:Groetzsch} below predicts that~$G'$ is $3$-edge-colorable.

\medskip

The computational complexity of the \textsc{3-Edge-Colorability}-problem for planar graphs remains open, while it is known to be \NP-complete already for $3$-regular, but not necessarily planar, graphs~\cite{Holyer1981_EdgeColoring}.
One can easily show that a planar subcubic graph is $3$-edge-colorable if and only if all of its blocks (inclusion-maximal $2$-connected subgraphs) are $3$-edge-colorable, i.e., \textsc{3-Edge-Colorability} reduces to the $2$-connected case.
A simple counting argument shows that a $2$-connected subcubic graph~$G$ with exactly one degree-$2$ vertex is not $3$-edge-colorable (independent of whether~$G$ is planar or not).
The following conjecture, attributed to Gr\"{o}tzsch by Seymour~\cite{Seymour1981_TuttesExtension}, states that in the case of planar graphs, this is the only obstruction.

\begin{conjecture}[Gr\"{o}tzsch, cf.~\cite{Seymour1981_TuttesExtension}]
    \label{conj:Groetzsch}
    If~$G$ is a $2$-connected planar graph of maximum degree~$\Delta(G) \leq 3$, then~$G$ is $3$-edge-colorable, unless it has exactly one vertex of degree~$2$.
\end{conjecture}

Note that if \cref{conj:Groetzsch} is true, \textsc{3-Edge-Colorability} would be in \P, as its condition is easy to check in linear time.
Thus a full answer to our initial question, \cref{quest:3-edge-colorability}, would most likely also resolve \cref{conj:Groetzsch}.

\medskip

Finally, let us also briefly discuss planar graphs of maximum degree larger than~$3$.
Vizing conjectured in 1965 that all planar graphs of maximum degree $\Delta \geq 6$ are $\Delta$-edge-colorable, proving it only for $\Delta \geq 8$~\cite{Vizing1965_CrititalGraphs}.
As of today, it is known that all planar graphs of maximum degree $\Delta \geq 7$ are $\Delta$-edge-colorable~\cite{Zhang2000_PlanarDegree7,Grunewald2000_ChromaticIndex,Sanders2001_PlanarMaxDegree7}, and optimal edge-colorings can be computed efficiently in these cases.
The case $\Delta = 6$ is still open, while for $\Delta = 3,4,5$ there are planar graphs of maximum degree $\Delta$ that are not $\Delta$-edge-colorable~\cite{Vizing1965_CrititalGraphs}, and at least for $\Delta = 4,5$ the \textsc{$\Delta$-Edge-Colorability}-problem is suspected to be \NP-complete for planar graphs~\cite{Chrobak1990_EdgeColoringAlgorithms}.

Generalizing \cref{conj:Groetzsch}, Seymour's Exact Conjecture~\cite{Seymour1981_TuttesExtension} states that every planar graph~$G$ is $\lceil \eta'(G) \rceil$-edge-colorable, where~$\eta'(G)$ denotes the fractional chromatic index of~$G$.
It is worth noting that Seymour's Exact Conjecture implies Vizing's Conjecture, as well as the Four-Color-Theorem; see e.g., the survey~\cite{Cao2019_EdgeColoringSurvey}.

\bibliographystyle{plainurl}
\bibliography{bibliography}

\begin{thebibliography}{10}

\bibitem{Appel1977_4Color1}
Kenneth Appel and Wolfgang Haken.
\newblock {Every planar map is four colorable. Part I: Discharging}.
\newblock {\em Illinois Journal of Mathematics}, 21(3):429--490, 1977.
\newblock \href {https://doi.org/10.1215/ijm/1256049011}
  {\path{doi:10.1215/ijm/1256049011}}.

\bibitem{Appel1977_4Color2}
Kenneth Appel and Wolfgang Haken.
\newblock {Every planar map is four colorable. Part II: Reducibility}.
\newblock {\em Illinois Journal of Mathematics}, 21(3):491--567, 1977.
\newblock \href {https://doi.org/10.1215/ijm/1256049012}
  {\path{doi:10.1215/ijm/1256049012}}.

\bibitem{Bhatt1987_GridSubgraph}
Sandeep~N. Bhatt and Stavros~S. Cosmadakis.
\newblock {The Complexity of Minimizing Wire Lengths in VLSI Layouts}.
\newblock {\em Information Processing Letters}, 25(4):263--267, 1987.
\newblock \href {https://doi.org/10.1016/0020-0190(87)90173-6}
  {\path{doi:10.1016/0020-0190(87)90173-6}}.

\bibitem{Biedl1997_4ConnectedTriangulation}
Therese Biedl, Goos Kant, and Michael Kaufmann.
\newblock {On Triangulating Planar Graphs under the Four-Connectivity
  Constraint}.
\newblock {\em Algorithmica}, 19(4):427--446, 1997.
\newblock \href {https://doi.org/10.1007/PL00009182}
  {\path{doi:10.1007/PL00009182}}.

\bibitem{Cao2019_EdgeColoringSurvey}
Yan Cao, Guantao Chen, Guangming Jing, Michael Stiebitz, and Bjarne Toft.
\newblock {Graph Edge Coloring: A Survey}.
\newblock {\em Graphs and Combinatorics}, 35(1):33--66, 2019.
\newblock \href {https://doi.org/10.1007/s00373-018-1986-5}
  {\path{doi:10.1007/s00373-018-1986-5}}.

\bibitem{Chen2023_AlmostLinearTime}
Li~Chen, Rasmus Kyng, Yang~P. Liu, Richard Peng, Maximilian~Probst Gutenberg,
  and Sushant Sachdeva.
\newblock Almost-linear-time algorithms for maximum flow and minimum-cost flow.
\newblock {\em Commun. ACM}, 66(12):85--–92, 2023.
\newblock \href {https://doi.org/10.1145/3610940} {\path{doi:10.1145/3610940}}.

\bibitem{Chrobak1990_EdgeColoringAlgorithms}
Marek Chrobak and Takao Nishizeki.
\newblock {Improved Edge-Coloring Algorithms for Planar Graphs}.
\newblock {\em Journal of Algorithms}, 11(1):102--116, 1990.
\newblock \href {https://doi.org/10.1016/0196-6774(90)90032-A}
  {\path{doi:10.1016/0196-6774(90)90032-A}}.

\bibitem{Cole2008_EdgeColoring}
Richard Cole and {\L}ukasz Kowalik.
\newblock {New Linear-Time Algorithms for Edge-Coloring Planar Graphs}.
\newblock {\em Algorithmica}, 50(3):351--368, 2008.
\newblock \href {https://doi.org/10.1007/s00453-007-9044-3}
  {\path{doi:10.1007/s00453-007-9044-3}}.

\bibitem{Cornuejols1988_GeneralFactors}
G{\'e}rard Cornu{\'e}jols.
\newblock {General Factors of Graphs}.
\newblock {\em Journal of Combinatorial Theory, Series B}, 45(2):185--198,
  1988.
\newblock \href {https://doi.org/10.1016/0095-8956(88)90068-8}
  {\path{doi:10.1016/0095-8956(88)90068-8}}.

\bibitem{deBerg2010_Planar3SAT}
Mark de~Berg and Amirali Khosravi.
\newblock {Optimal Binary Space Partitions in the Plane}.
\newblock In My~T. Thai and Sartaj Sahni, editors, {\em Computing and
  Combinatorics (COCOON 2010)}, volume 6196 of {\em Lecture Notes in Computer
  Science}, pages 216--225, 2010.
\newblock \href {https://doi.org/10.1007/978-3-642-14031-0_25}
  {\path{doi:10.1007/978-3-642-14031-0_25}}.

\bibitem{Fraysseix1994_Augmentation}
Hubert de~Fraysseix and Patrice Ossona~de Mendez.
\newblock {Regular Orientations, Arboricity, and Augmentation}.
\newblock In Roberto Tamassia and Ioannis~G. Tollis, editors, {\em Graph
  Drawing (GD 1994)}, volume 894 of {\em Lecture Notes in Computer Science},
  pages 111--118, 1994.
\newblock \href {https://doi.org/10.1007/3-540-58950-3_362}
  {\path{doi:10.1007/3-540-58950-3_362}}.

\bibitem{DiBattista1996_SPQR}
Giuseppe Di~Battista and Roberto Tamassia.
\newblock {On-Line Planarity Testing}.
\newblock {\em SIAM Journal on Computing}, 25(5):956--997, 1996.
\newblock \href {https://doi.org/10.1137/S0097539794280736}
  {\path{doi:10.1137/S0097539794280736}}.

\bibitem{DiGiacomo2010_HamiltonianAugmentation}
Emilio Di~Giacomo and Giuseppe Liotta.
\newblock {The Hamiltonian Augmentation Problem and Its Applications to Graph
  Drawing}.
\newblock In Md.~Saidur Rahman and Satoshi Fujita, editors, {\em WALCOM:
  Algorithms and Computation (WALCOM 2010)}, volume 5942 of {\em Lecture Notes
  in Computer Science}, pages 35--46, 2010.
\newblock \href {https://doi.org/10.1007/978-3-642-11440-3_4}
  {\path{doi:10.1007/978-3-642-11440-3_4}}.

\bibitem{Eswaran1976_AugmentationProblems}
Kapali~P. Eswaran and R.~Endre Tarjan.
\newblock {Augmentation Problems}.
\newblock {\em SIAM Journal on Computing}, 5(4):653--665, 1976.
\newblock \href {https://doi.org/10.1137/0205044} {\path{doi:10.1137/0205044}}.

\bibitem{Goetze2022-CubicSubgraphsESA}
Miriam Goetze, Paul Jungeblut, and Torste Ueckerdt.
\newblock {Efficient Recognition of Subgraphs of Planar Cubic Bridgeless
  Graphs}.
\newblock In Shiri Chechik, Gonzalo Navarro, Eva Rotenberg, and Grzegorz
  Herman, editors, {\em 30th Annual European Symposium on Algorithms (ESA
  2022)}, volume 244 of {\em Leibniz International Proceedings in Informatics
  (LIPIcs)}, pages 62:1--62:15, 2022.
\newblock \href {https://doi.org/10.4230/LIPIcs.ESA.2022.62}
  {\path{doi:10.4230/LIPIcs.ESA.2022.62}}.

\bibitem{Goetze2024-CubicSubgraphsEuroCG}
Miriam Goetze, Paul Jungeblut, and Torsten Ueckerdt.
\newblock {Recognition Complexity of Subgraphs of 2- and 3-Connected Planar
  Cubic Graphs}.
\newblock In {\em Proceedings of the 40th European Workshop on Computational
  Geometry (EuroCG 2024)}, pages 40:1--40:8, 2024.

\bibitem{Grunewald2000_ChromaticIndex}
Stefan Gr{\"u}newald.
\newblock {\em {Chromatic Index Critical Graphs and Multigraphs}}.
\newblock PhD thesis, Universit\"{a}t Bielefeld, 2000.
\newblock URL: \url{https://pub.uni-bielefeld.de/record/2303675}.

\bibitem{Gutwenger2001_SPQRLinear}
Carsten Gutwenger and Petra Mutzel.
\newblock {A Linear Time Implementation of SPQR-Trees}.
\newblock In Joe Marks, editor, {\em Graph Drawing (GD 2000)}, volume 1984 of
  {\em Lecture Notes in Computer Science}, pages 77--90, 2001.
\newblock \href {https://doi.org/10.1007/3-540-44541-2_8}
  {\path{doi:10.1007/3-540-44541-2_8}}.

\bibitem{Hartmann2015_RegularAugmentation}
Tanja Hartmann, Jonathan Rollin, and Ignaz Rutter.
\newblock Regular augmentation of planar graphs.
\newblock {\em Algorithmica}, 73(2):306--370, 2015.
\newblock \href {https://doi.org/10.1007/s00453-014-9922-4}
  {\path{doi:10.1007/s00453-014-9922-4}}.

\bibitem{Holyer1981_EdgeColoring}
Ian Holyer.
\newblock {The NP-Completeness of Edge-Coloring}.
\newblock {\em SIAM Journal on Computing}, 10(4):718--720, 1981.
\newblock \href {https://doi.org/10.1137/0210055} {\path{doi:10.1137/0210055}}.

\bibitem{Kant1991_PlanarAugmentation}
Goos Kant and Hans~L. Bodlaender.
\newblock {Planar Graph Augmentation Problems}.
\newblock In Frank Dehne, J{\"o}rg-R{\"u}diger Sack, and Nicola Santoro,
  editors, {\em Algorithms and Data Structures (WADS 1991)}, volume 519 of {\em
  Lecture Notes in Computer Science}, pages 286--298, 1991.
\newblock \href {https://doi.org/10.1007/BFb0028270}
  {\path{doi:10.1007/BFb0028270}}.

\bibitem{Kant1997_TriangulatingMaxDegree}
Goos Kant and Hans~L. Bodlaender.
\newblock {Triangulating Planar Graphs While Minimizing the Maximum Degree}.
\newblock {\em Information and Computation}, 135(1):1--14, 1997.
\newblock \href {https://doi.org/10.1006/inco.1997.2635}
  {\path{doi:10.1006/inco.1997.2635}}.

\bibitem{Kratochvil2012_Planar3Tree}
Jan Kratochv{\'{\i}}l and Michal Vaner.
\newblock A note on planar partial 3-trees.
\newblock arXiv preprint, 2012.
\newblock URL: \url{http://arxiv.org/abs/1210.8113}.

\bibitem{Lovasz1972_Factorization}
L{\'a}szl{\'o} Lov{\'a}sz.
\newblock {The Factorization of Graphs. II}.
\newblock {\em Acta Mathematica Academiae Scientiarum Hungarica},
  23(1--2):223--246, 1972.
\newblock \href {https://doi.org/10.1007/BF01889919}
  {\path{doi:10.1007/BF01889919}}.

\bibitem{Rutter2012_Wolff_2EdgePlanarAugmentation}
Ignaz Rutter and Alexander Wolff.
\newblock Augmenting the connectivity of planar and geometric graphs.
\newblock {\em Journal of Graph Algorithms and Applications}, 16(2):599–628,
  2012.
\newblock \href {https://doi.org/10.7155/jgaa.00275}
  {\path{doi:10.7155/jgaa.00275}}.

\bibitem{Sanders2001_PlanarMaxDegree7}
Daniel~P. Sanders and Yue Zhao.
\newblock {Planar Graphs of Maximum Degree Seven are Class I}.
\newblock {\em Journal of Combinatorial Theory, Series B}, 83(2):201--212,
  2001.
\newblock \href {https://doi.org/10.1006/jctb.2001.2047}
  {\path{doi:10.1006/jctb.2001.2047}}.

\bibitem{Sebo1993_Antifactors}
Andr{\'{a}}s Seb{\H{o}}.
\newblock {General Antifactors of Graphs}.
\newblock {\em Journal of Combinatorial Theory, Series B}, 58(2):174--184,
  1993.
\newblock \href {https://doi.org/10.1006/jctb.1993.1035}
  {\path{doi:10.1006/jctb.1993.1035}}.

\bibitem{Seymour1981_TuttesExtension}
Paul~D. Seymour.
\newblock {On Tutte's Extension of the Four-Colour Problem}.
\newblock {\em Journal of Combinatorial Theory, Series B}, 31(1):82--94, 1981.
\newblock \href {https://doi.org/10.1016/S0095-8956(81)80013-5}
  {\path{doi:10.1016/S0095-8956(81)80013-5}}.

\bibitem{Tait1880_Bridgeless}
Peter~G. Tait.
\newblock {10. Remarks on the previous Communication}.
\newblock {\em Proceedings of the Royal Society of Edinburgh}, 10:729–--731,
  1880.
\newblock \href {https://doi.org/10.1017/S0370164600044643}
  {\path{doi:10.1017/S0370164600044643}}.

\bibitem{Tarjan1972_Blocks}
Robert~E. Tarjan.
\newblock {Depth-First Search and Linear Graph Algorithms}.
\newblock {\em SIAM Journal on Computing}, 1(2):146--160, 1972.
\newblock \href {https://doi.org/10.1137/0201010} {\path{doi:10.1137/0201010}}.

\bibitem{Tarjan1974_Bridges}
Robert~E. Tarjan.
\newblock {A Note on Finding the Bridges of a Graph}.
\newblock {\em Information Processing Letters}, 2(6):160--161, 1974.
\newblock \href {https://doi.org/10.1016/0020-0190(74)90003-9}
  {\path{doi:10.1016/0020-0190(74)90003-9}}.

\bibitem{Vizing1965_CrititalGraphs}
Vadim~G. Vizing.
\newblock {Critical Graphs with Given Chromatic Class}.
\newblock {\em Akademiya Nauk SSSR. Sibirskoe Otdelenie. Institut Matematiki.
  Diskretny\u{\i} Analiz. Sbornik Trudov}, pages 9--17, 1965.
\newblock In Russian.

\bibitem{Whitney1933_UniqueEmbedding}
Hassler Whitney.
\newblock {2-Isomorphic Graphs}.
\newblock {\em American Journal of Mathematics}, 55(1):245--254, 1933.
\newblock \href {https://doi.org/10.2307/2371127} {\path{doi:10.2307/2371127}}.

\bibitem{Zhang2000_PlanarDegree7}
Limin Zhang.
\newblock {Every Planar Graph with Maximum Degree 7 Is of Class 1}.
\newblock {\em Graphs and Combinatorics}, 16(4):467--495, 2000.
\newblock \href {https://doi.org/10.1007/s003730070009}
  {\path{doi:10.1007/s003730070009}}.

\end{thebibliography}

\end{document}